\pdfoutput=1
\documentclass[aip,jmp,preprint,groupedaddress,
floatfix,amsmath,amssymb,notitlepage]{revtex4-1}

\usepackage{lmodern}
\usepackage{avant}
\usepackage{courier}
\usepackage{titlesec}
\usepackage[latin9]{inputenc}
\usepackage[a4paper]{geometry}
\usepackage{stmaryrd}
\SetSymbolFont{stmry}{bold}{U}{stmry}{m}{n}
\usepackage[english]{babel}
\usepackage{amsthm}
\usepackage{chngcntr}
\usepackage{color}
\usepackage{subfigure}
\usepackage{graphicx}
\usepackage[unicode=true,pdfusetitle,
 bookmarks=true,bookmarksnumbered=false,bookmarksopen=false,
 breaklinks=false,citecolor=blue,linkcolor=blue,pdfborder={0 0 1},backref=false,colorlinks=true]
 {hyperref}

\makeatletter

\numberwithin{equation}{section}
\numberwithin{figure}{section}
\numberwithin{table}{section}
\theoremstyle{plain}

  \theoremstyle{plain}
  \newtheorem{prop}{\protect\propositionname}
  \theoremstyle{plain}
  \newtheorem{lem}{\protect\lemmaname}
  \theoremstyle{plain}
  \newtheorem{cor}{\protect\corollaryname}

\providecommand{\keywords}[1]{\textsc{\textit{Keywords---}} #1}

\makeatother

  \providecommand{\corollaryname}{Corollary}
  \providecommand{\lemmaname}{Lemma}
  \providecommand{\propositionname}{Proposition}
\providecommand{\theoremname}{Theorem}

\begin{document}

\title{Zeros and Amoebas of Partition Functions}

\author{M. Angelelli}
\email{mario.angelelli@le.infn.it}
\author{B. Konopelchenko}
 \email{boris.konopeltchenko@unisalento.it}
\affiliation{
 Department of Mathematics and Physics\\
 ``Ennio De Giorgi'', University of Salento and sezione INFN,\\
 Lecce, 73100, Italy
}

\begin{abstract}

Singular sectors $\mathcal{Z}_{\mathrm{sing}}$ (loci of zeros) for real-valued non-positively defined partition functions $\mathcal{Z}$ of $n$ variables are studied. It is shown that $\mathcal{Z}_{\mathrm{sing}}$ have a stratified structure and each stratum is a set of certain hypersurfaces in $\mathbb{R}^n$. The concept of statistical amoebas is introduced and their properties are studied. Relation with algebraic amoebas is discussed. Tropical limit of statistical amoebas is considered too.

\end{abstract}


\maketitle
\counterwithout{figure}{section}

\section{Introduction \label{sec: Introduction}}

The partition function is a key object in various branches of physics.
In statistical physics, due to the relation with the free energy $F=-k_{B}T\ln\mathcal{Z}$,
all basic thermodynamic characteristics of a macroscopic system are
encoded in the partition function 
\begin{equation}
\mathcal{Z}=\sum_{n}g_{n}\cdot e^{-\frac{E_{n}}{k_{B}T}}
\label{eq: partition function}
\end{equation}
where $\{E_{n}\}$ is the energy spectrum, $g_{n}$ is the degeneracy
of the $n$-th state, $T$ is the temperature and $k_{B}$ is the
Boltzmann constant (see e.g. \cite{LL1980,Huang1963}). In equilibrium
$\mathcal{Z}$ is finite and positive, otherwise it is an indication
of instability of the system. 

Zeros of the partition function as a function of physical parameters (temperature,
magnetic field \textit{etc.}) are of particular interest since at
such points the free energy becomes singular and, hence, the system changes
of state, e.g. exhibits a phase transition \cite{LL1980,Huang1963}.
Since the seminal papers of Lee and Yang \cite{LeeYang1952a,LeeYang1952b}
it is known that for usual systems (with finite and positive $g_{n}$
and real $E_{n}$) zeros of the partition functions lie in the complex
plane (see e.g. \cite{Huang1963, Fisher1965, KG1971, FS1971, Ruelle1973, Derrida1981, Lieb1981, Pearson1982, Derrida1983, Borgs1990,Biskup2004}). These
results have led to the intensive study of the partition function's
zero sets and associated phase transitions in the complex plane of
physical parameters for a number of models in statistical physics,
including those subject to quantum dynamics (see e.g. \cite{Wei2012,Takahashi2012,Wei2014,Peng2015}). 

Study of unstable or metastable states is another branch of statistical
physics where complex-valued partition functions naturally arise \cite{Langer1969,Newman1980}.
Formally, for an unstable state the energy $E_{n}$ is complex $E_{n}=E_{n,0}+\mathrm{i}\Delta E_{n}$
with width $\Delta E_{n}$ and, hence, the partition function is complex-valued
too. Wide classes of macroscopic systems like spin-glasses and other
geometrically or dynamically frustrated systems have such peculiarity
(see e.g. \cite{Parisi1980,MPV1987,Nishimori1988,Derrida1991,Bhanot1993,Matsuda2008,Obuchi2012}). 

In all these cases the situation when partition function's zeros are
real is of the greatest interest. The study of properties of macroscopic
systems, in particular, structure of equilibrium and unstable domains,
is simplified if the partition function is real-valued for all values
of parameters. Such situation is realisable, for example, for spin-glasses
and frustrated systems with different temperatures $T_{n}$ of microsystems
(microbasins) if the widths $\Delta E_{n}$ of energy levels obey
the condition ${\displaystyle \frac{\Delta E_{n}}{k_{B}T_{n}}=\ell_{n}\pi}$,
$n=1,2,3,\dots$ where $\ell_{n}$ are integers. Terms with odd $\ell_{n}$
acquire the factor $-1$ and, hence, the partition function $\mathcal{Z}$
is of the form (\ref{eq: partition function}) and real-valued, but
with $g_{n}$ assuming both positive and negative values. Negative
degeneracies of energy levels can be interpreted as the contribution
from sort of holes in spectrum. Formally, negativity of $g_{n}$ is
closely connected with the concepts of negative probability and negative
membership functions widely discussed in literature (see e.g. \cite{Feynman1987,Blizard1990,Burgin2010}
and references therein). 

This paper is devoted to the study of partition functions of such
a type, more precisely, those of the form 
\begin{equation}
\mathcal{Z}(\boldsymbol{g};\boldsymbol{x}):=\sum_{\alpha=1}^{N}g_{\alpha}\cdot e^{f_{\alpha}(x_{1},\dots,x_{n})}
\label{eq: signed partition function}
\end{equation}
where factors $g_{\alpha}$ take values $1$ or $-1$, $x_{1},\dots,x_{n}$
are real variables and $f_{\alpha}(\boldsymbol{x})$ are linear real-valued functions. Main attention
is paid to an analysis of the singular sector $\mathcal{Z}_{\mathrm{sing}}$
(locus of zeros) of the partition function (\ref{eq: signed partition function}),
its stratification and structure of stability ($\mathcal{Z}(\boldsymbol{g};\boldsymbol{x})>0$)
and zero confinement domains in the space $\mathbb{R}^{n}$ of parameters
$(x_{1},\dots,x_{n})$. 

Singular sector $\mathcal{Z}_{\mathrm{sing}}$ admits a natural stratification
\begin{equation}
\mathcal{Z}_{\mathrm{sing}}=\bigcup_{k=1}^{\lfloor\frac{N}{2}\rfloor}\mathcal{Z}_{\mathrm{sing},k}:
\label{eq: k-stratification singular locus}
\end{equation}
the stratum $\mathcal{Z}_{\mathrm{sing},k}$ is composed by all hypersurfaces
given by ${\displaystyle \binom{N}{k}}$ equations 
\begin{equation}
\mathcal{Z}_{k}(\mathcal{I};\boldsymbol{x}):=-\sum_{\alpha\in\mathcal{I}}e^{f_{\alpha}(\boldsymbol{x})}+\sum_{\beta\notin\mathcal{I}}e^{f_{\beta}(\boldsymbol{x})}=0,\quad k=1,\dots,\left\lfloor\frac{N}{2}\right\rfloor\label{eq: k-subset singular component}
\end{equation}
where $(\mathcal{I},[N]\backslash\mathcal{I})$ is any $2$-partition of
the set $[N]:=\{1,2,\dots,N\}$ with cardinality
$\#\mathcal{I}=k$. These hypersurfaces for the $k$-stratum are
contained in a certain domain in $\mathbb{R}^{n}$ refered as the
zero confinement domain $ZCD_{k}$. This domain is divided by hypersurfaces
(\ref{eq: k-subset singular component}) into a number of subdomains
$ZCD_{k;\delta}$ at which each of the $\binom{N}{k}$ functions $\mathcal{Z}_{k}$
in (\ref{eq: k-subset singular component}) is positive or negative. This allows to associate with each of these subdomains a
set of $\binom{N}{k}$ number $1$ or $-1$ that can be viewed as the
state of the system of $\binom{N}{k}$ ``spins'' which take values
$1$ or $-1$. 

The domain $\mathcal{D}_{k+}$, where all functions (\ref{eq: k-subset singular component})
are positive, is the stability (equilibrium) domain. The union $\mathcal{A}_{k}:=\mathcal{D}_{k+}\cup ZCD_{k}$
is called the statistical $k$-amoeba. A $k$-statistical amoeba is composed
by the stable nucleus $\mathcal{D}_{k+}$ and intermittent shell $ZCD_{k}$
with varying degree of instability (number of signs $-1$). The complement
$\mathcal{D}_{k-}$ of the $k$-amoeba in $\mathbb{R}^{n}$ is the
domain with maximum number of signs $-1$, i.e. the domain of maximal
instability. 

The domains $\mathcal{D}_{k+}$, $ZCD_{k}$ and statistical $k$-amoebas
exhibit a simple inclusion property, for instance, $\mathcal{D}_{k+}\supseteq\mathcal{D}_{\hat{k}+}$
and $\mathcal{A}_{k}\supseteq\mathcal{A}_{\hat{k}}$ if ${\displaystyle 1\leq k<\hat{k}<\frac{N}{2}}$. 

Statistical $1$-amoebas $\mathcal{A}_{1}$ coincide with the so-called
non-lopsided amoebas $\mathcal{LA}$ introduced in algebraic geometry
\cite{Purbhoo2008}. Analogs of higher statistical $k$-amoebas ($k\geq2$)
seem to be not studied in algebraic geometry. 

Tropical limits of statistical $k$-amoebas are considered. It is
shown that all $k$-amoebas collapse into the same set of piecewise hyperplanes in $\mathbb{R}^{n}$
coinciding with that of tropical limit of $\mathcal{A}_{1}$. 

It should be noted that partition functions depending on several
variables (multidimensional energy spectrum or several Hamiltonians)
have been considered in \cite{Tsikh2009}, \cite{Kapranov2011}
and \cite{Passare2012}. However these papers have addressed only
the case (\ref{eq: signed partition function}) with all positive factors
$g_{\alpha}$. 

In a completely different setting zeros of superpositions of the form
(\ref{eq: signed partition function}) ($\tau$-functions) with positive
and negative factors $g_{\alpha}$ and very particular linear functions
$f_{\alpha}$ arise within an analysis of singular
solutions of integrable equations (see e.g. \cite{Adler1980}, \cite{Kodama2006}). 

The paper is organized as follows. General definitions of strata of
the singular sector and some concrete examples for the first stratum
$\mathcal{Z}_{\mathrm{sing},1}$ are given in section \ref{sec: Singular sector of partition function}.
Higher strata and properties of zero loci hypersurfaces are considered
in section \ref{sec: Higher strata: zero loci}. Section \ref{sec: Higher strata: general properties}  is devoted to the study of some general properties of equilibrium and zero confinement domains. Statistical amoebas and their
relation to algebraic amoebas are discussed in section \ref{sec: Statistical amoebas vs. algebraic amoebas}.
Next section \ref{sec: Structure of ZCDk domains} is devoted to an
analysis of the structure and properties of $ZCD_{k}$ domains and
associated statistical systems of ``spins''. Tropical limits of
statistical amoebas are considered in section \ref{sec: Tropical limit and tropical zeros}. In conclusion some peculiarities of partition function (\ref{eq: signed partition function}) with nonlinear functions $f_{\alpha}$ are noted.

\section{Singular sector of partition function \label{sec: Singular sector of partition function}}

So we will consider the family of partition functions of the form
\begin{equation}
\mathcal{Z}=\sum_{\alpha=1}^{N}g_{\alpha}\cdot e^{f_{\alpha}(\boldsymbol{x})}\label{eq: signed partition function, 2}
\end{equation}
where ${\displaystyle f_{\alpha}(\boldsymbol{x})=b_{\alpha}+\sum_{i=1}^{n}a_{\alpha i}x_{i}}$,
all variables $g_{\alpha}$, $x_{i}$ and all functions $f_{\alpha}$
are real. Since $g_{\alpha}e^{b_{\alpha}}=\text{sign}(g_{\alpha})\cdot e^{\log|g_{\alpha}|+b_{\alpha}}$ one can consider only the case $g_{\alpha}\in\{+1,-1\}$.  

The space $V$ of parameters $g_{1},\dots,g_{N},x_{1},\dots,x_{n}$ admits
the stratification 
\begin{equation}
V=\bigcup_{\alpha=0}^{N}V_{\alpha}\label{eq: stratification variable space}
\end{equation}
where $V_{\alpha}$ is the union of subspaces of $V$ with $\alpha$ many 
negative $g_{\beta}$. For instance, ${\displaystyle V_{2}=\bigcup_{1\leq\alpha<\beta\leq N,}V_{2,\{\alpha\beta\}}}$
where $\displaystyle V_{2,\{\alpha\beta\}}=\left\{ (g_{1},\dots,g_{N};x_{1},\dots,x_{n}):\,g_{\alpha}=g_{\beta}=-1,\right.$ \\ $\displaystyle \left.\gamma\neq\alpha,\beta\Rightarrow g_{\gamma}>0\right\}$,
$\alpha\neq\beta$, $\alpha,\beta=1,\dots,N$. Inversion $P$ of all
$g_{\alpha}$: $Pg_{\alpha}=-g_{\alpha}$, $\alpha=1,\dots N$, acts
on strata $V_{\alpha}$ as $PV_{\alpha}=V_{N-\alpha}$. 

Accordingly, singular sector $\mathcal{Z}_{\mathrm{sing}}$ of partition
function also admits the stratification 
\begin{equation}
\mathcal{Z}_{\mathrm{sing}}=\bigcup_{k=1}^{N-1}\mathcal{Z}_{\mathrm{sing},k}\label{eq: stratitifaction singular sector}
\end{equation}
where $\mathcal{Z}_{\mathrm{sing},k}$ are subspaces of $V_{k}$ for
which $\displaystyle \mathcal{Z}|_{V_{k}}=0$. Subspaces $V_{0}$
and $V_{N}$ are obviously regular and connected by inversion $P$.
First singular stratum $\mathcal{Z}_{\mathrm{sing},1}$ is the union
of the $N$ hypersurfaces defined by the equations 
\begin{equation}
\mathcal{Z}_{1}(\{\alpha\};\boldsymbol{x}):=\sum_{\beta=1}^{N}g_{\alpha(\beta)}e^{f_{\beta}(\boldsymbol{x})}=0,\quad\alpha=1,\dots,N\label{eq: hypersurfaces equation 1-stratum}
\end{equation}
with $g_{\alpha(\alpha)}=-1$ and $g_{\alpha(\beta)}=1$, $\beta\neq\alpha$.
Geometric characteristics of such hypersurfaces have been studied
in the paper \cite{AK2016}. 
Cases of linear functions $f_{\alpha}(\boldsymbol{x})$ were referred
in \cite{AK2016} as ideal statistical hypersurfaces. General statistical
hypersurfaces considered in \cite{AK2016} were associated with nonlinear
functions $f_{\alpha}$ while super-ideal case corresponds to $N=n$
and $f_{\alpha}(\boldsymbol{x})\equiv x_{\alpha}$. 

Generically the stratum
$\mathcal{Z}_{\mathrm{sing},1}$ can be composed by $N$ hypersurfaces.
In addition, it is easy to see that hypersurfaces given by $\mathcal{Z}_{1}(\{\alpha\};\boldsymbol{x})=0$
with different $\alpha$ do not intersect at finite values of $x_{1},\dots,x_{n}$
and $N\geq3$. Indeed, if there exists $\alpha\neq\beta$ such that
$\{\boldsymbol{x}:\mathcal{Z}_{1}(\{\alpha\};\boldsymbol{x})=0\}\cap\{\boldsymbol{x}:\mathcal{Z}_{1}(\{\beta\};\boldsymbol{x})=0\}$ is not empty, then
there exists $\boldsymbol{x}_{0}$ in $\mathbb{R}^{n}$ such that
${\displaystyle e^{f_{\beta}(\boldsymbol{x}_{0})}-e^{f_{\alpha}(\boldsymbol{x}_{0})}=\sum_{\gamma\notin\{\alpha,\beta\}}e^{f_{\gamma}(\boldsymbol{x}_{0})}=e^{f_{\alpha}(\boldsymbol{x}_{0})}-e^{f_{\beta}(\boldsymbol{x}_{0})}}$,
that is $f_{\beta}(\boldsymbol{x}_0)=f_{\alpha}(\boldsymbol{x}_0)$, so
one has ${\displaystyle \sum_{\gamma\neq\alpha,\beta}e^{f_{\gamma}(\boldsymbol{x}_{0})}=0}$
which is impossible if $N\geq3$. 

The form of hypersurfaces which compose the sector $\mathcal{Z}_{\mathrm{sing},k}$
depends on $N$, $n$ and the choice of functions $f_{\alpha}(\boldsymbol{x})$.
At the simplest case of $n=1$ the stratum $\mathcal{Z}_{\mathrm{sing},1}$ is composed, in general,
by at most $N$ points defined by the equations 
\begin{equation}
\mathcal{Z}_{1}(\{\alpha\};x_{1})=\sum_{\beta=1}^{N}g_{\alpha(\beta)}e^{a_{\beta}x_{1}+b_{\beta}}=0,\quad \alpha=1,\dots,N\label{eq: hypersurfaces 1-stratum n=1}
\end{equation}
with $g_{\alpha(\alpha)}=-1$ and $g_{\alpha(\beta)}=1$, $\beta\neq\alpha$.
In the case $N=2$ the sector $\mathcal{Z}_{\mathrm{sing},1}$ contains
only one hypersurface defined by the equation 
\begin{equation}
\mathcal{Z}_{1}(\{1\};\boldsymbol{x})=-e^{f_{1}(x_{1},\dots,x_{n})}+e^{f_{2}(x_{1},\dots,x_{n})}=0.\label{eq: hypersurfaces 1-stratum N=2}
\end{equation}
It is the hyperplane in $\mathbb{R}^{n}$ given by the equation 
\begin{equation}
b_{1}-b_{2}+\sum_{i=1}^{n}(a_{1i}-a_{2i})x_{i}=0.\label{eq: hypersurface 1-stratum N=2, 2}
\end{equation}
For $N\geq3$ and $n=2$ one has a family of curves on the plane $(x_{1},x_{2})=:(x,y)$.
For example, at the choice $f_{1}\equiv 0$, $f_{2}\equiv x$, $f_{3}\equiv y$
the sector $\mathcal{Z}_{\mathrm{sing},1}$ is composed by three curves
shown in figure \ref{fig: examples of 1-strata, a} (see also \cite{Mikhalkin2013})
\begin{figure}[tph]
\centering
\subfigure[$f_{1}\equiv0$, $f_{2}\equiv x$, $f_{3}\equiv y$. ]{
\includegraphics[scale=0.4]{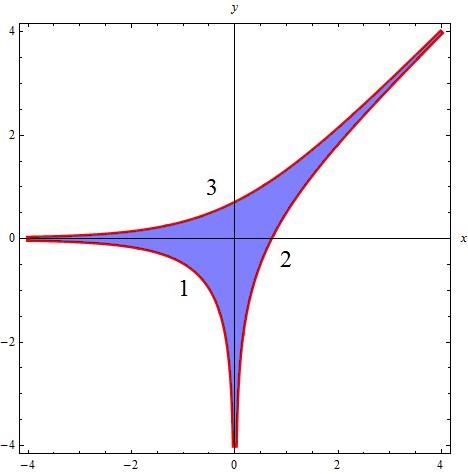}\label{fig: examples of 1-strata, a}}\hfill{}\subfigure[$f_{1}\equiv0$, $f_{2}\equiv3x$, $f_{3}\equiv3y$, $f_{4}\equiv x+y+\ln6$.]{

\includegraphics[scale=0.4]{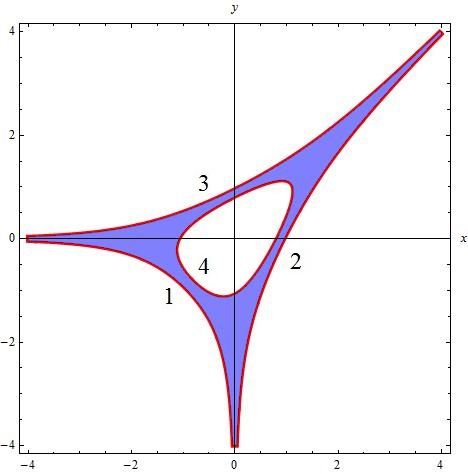}\label{fig: examples of 1-strata, b}}
\caption{Examples of $1$-strata (red curves). }
\label{fig: examples of 1-strata}
\end{figure}
where the curves $1$, $2$ and $3$ are given by the equations 
\begin{equation}
\begin{array}{c}
\mathcal{Z}_{1}(\{1\};x,y)\equiv -1+e^{x}+e^{y}=0,\\
\mathcal{Z}_{1}(\{2\};x,y)\equiv 1-e^{x}+e^{y}=0,\\
\mathcal{Z}_{1}(\{3\};x,y)\equiv 1+e^{x}-e^{y}=0.
\end{array}
\label{eq: components 1-stratum, ex1}
\end{equation}
Note that at $|x|,|y|\rightarrow\infty$ the curves $1$ and $2$
tend to the ray $x=0$, $y<0$, the curves $2$ and $3$ tend to the
ray $x=y$, $x>0$ while the curves $3$ and $1$ tend to the ray
$y=0$, $x<0$. An example with a different homotopy and
a bounded closed component (curve $4$) is given by choice $f_{1}\equiv 0$,
$f_{2}\equiv 3x$, $f_{3}\equiv 3y$, $f_{4}\equiv x+y+\ln 6$ and it is shown in figure \ref{fig: examples of 1-strata, b}. 

The stratum $\mathcal{Z}_{\mathrm{sing},1}$ at $n=2$
is composed by at most $N$ curves. For particular choice
of functions $f_{\alpha}(\boldsymbol{x})$ this number can be smaller than $N$ when some equations in (\ref{eq: hypersurfaces equation 1-stratum})
define the empty set. To illustrate this let us consider the following
examples with $n=2$: if one chooses 
\begin{equation}
f_{\alpha}(x,y)\equiv\cos\left(\frac{2\pi(\alpha-1)}{N}\right)\cdot x+\sin\left(\frac{2\pi(\alpha-1)}{N}\right)\cdot y,\quad\alpha=1,\dots,N,\label{eq: example family maximal number 1-components}
\end{equation}
then one gets $N$ curves by symmetry. For instance, at $N=6$ one
gets $6$ curves presented in figure \ref{fig: example 1-stratum highly symmetric}.
If, instead, one takes 
\begin{equation}
\displaystyle f_{\alpha}(x,y)\equiv(\alpha-1)\cdot x+(N-\alpha)\cdot y, \quad\alpha=1,\dots,N,\label{eq: example family highly hiding number 1-components}
\end{equation}
then for all $2\leq\gamma\leq N-1$ the corresponding locus is empty,
i.e. $\{\mathcal{Z}_{1}(\{\gamma\})=0\}=\emptyset$. Indeed, if $x\leq y$ then $e^{(\gamma-1)\cdot x}\leq e^{(\gamma-1)\cdot y}$
thus ${\displaystyle e^{f_{\gamma}(x,y)}\leq e^{f_{1}(x,y)}<\sum_{\beta\neq\gamma}e^{f_{\beta}(x,y)}}$.
Similarly, If $y\leq x$ then ${\displaystyle e^{f_{\gamma}(x,y)}\leq e^{f_{N}(x,y)}<\sum_{\beta\neq\gamma}e^{f_{\beta}(x,y)}}$.
The only two visible curves are given by $\mathcal{Z}_{1}(\{1\};x,y)=0$
and $\mathcal{Z}_{1}(\{N\};x,y)=0$: they are straight lines with
slope $1$ passing through $(\pm\ln\chi,0)$ respectively, where $\chi>0$ is uniquely defined by $\displaystyle -1+\sum_{\alpha=1}^{5}\chi^{\alpha}=0$.
See figure \ref{fig: example 1-stratum highly degenerate} for the
case at $N=6$. 
\begin{figure}[tph]
\centering
\subfigure[$f_{\alpha}\equiv\cos\left(\frac{\pi(\alpha-1)}{3}\right)x+\sin\left(\frac{\pi(\alpha-1)}{3}\right)y$,
$\alpha=1,\dots6$: all $6$ components are visible. ]{\includegraphics[scale=0.4]{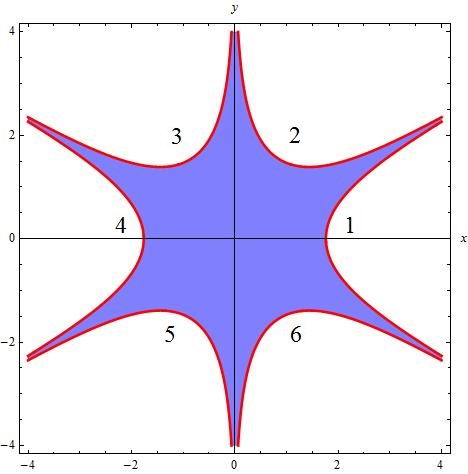}\label{fig: example 1-stratum highly symmetric}}\hfill{}\subfigure[${\displaystyle f_{\alpha}(x,y)\protect:=(\alpha-1)\cdot x+(6-\alpha)\cdot y}$,
$\alpha=1,\dots,6$. ]{\includegraphics[scale=0.4]{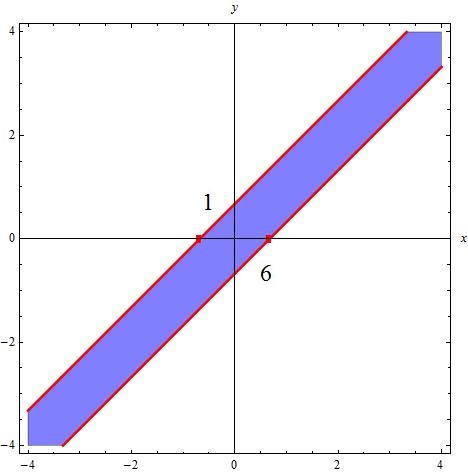}\label{fig: example 1-stratum highly degenerate}}
\caption{Extremal behaviors of number of visible curves. }
\label{fig: examples number curves}
\end{figure}

One has an intermediate case at $f_{1}(x,y)\equiv0,$
$f_{2}(x,y)\equiv3x$, $f_{3}(x,y)\equiv3y$, $f_{4}(x,y)\equiv x+y+\ln6$,
$f_{5}(x,y)\equiv2x+y+\ln11$, $f_{6}(x,y)\equiv x+3y+\ln4$. The stratum $\mathcal{Z}_{\mathrm{sing},1}$ is composed by $5$
curves given in figure \ref{fig: intermediate number curves 1-stratum}.
\begin{figure}[tph]
\centering{
\includegraphics[scale=0.5]{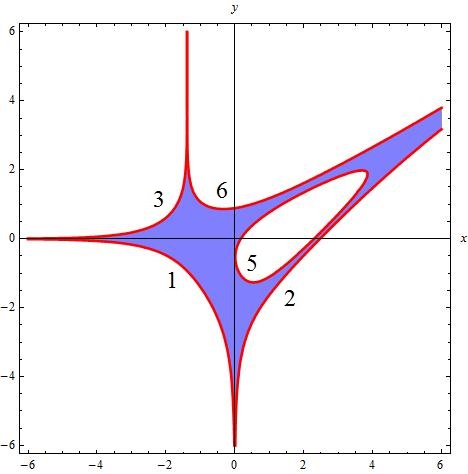}
\caption{$1$-stratum for the choice {\small{}$f_{1}(x,y)\equiv 0,$ $f_{2}(x,y)\equiv 3x$,
$f_{3}(x,y)\equiv 3y$, $f_{4}(x,y)\equiv x+y+\ln6$, $f_{5}(x,y)\equiv 2x+y+\ln11$,
$f_{6}(x,y)\equiv x+3y+\ln4$. }}
\label{fig: intermediate number curves 1-stratum}
}
\end{figure}
One can check directly that the set of solutions of the equation 
\begin{equation}
\mathcal{Z}_{1}(\{4\};x,y)\equiv 1+e^{3x}+e^{3y}-6\cdot e^{x+y}+11\cdot e^{2x+y}+4e^{x+3y}=0\label{eq: invisible component main example}
\end{equation}
is empty. Indeed, if $\mathcal{Z}_{1}(\{4\};x,y)=0$ then $x+y+\ln6>2x+y+\ln11$,
that is $x<\ln6-\ln11<0$. From the arithmetic-geometric means inequality
one gets ${\displaystyle \frac{1}{2}+\frac{1}{2}+11\cdot e^{2x+y}+4\cdot e^{x+3y}\geq}$
${\displaystyle 4\left[\frac{1}{2}\cdot\frac{1}{2}\cdot(11\cdot e^{2x+y})\cdot(4\cdot e^{x+3y})\right]^{\frac{1}{4}}=4\cdot11^{\frac{1}{4}}\cdot e^{\frac{3x+4y}{4}}}$.
But $4\cdot11^{\frac{1}{4}}>6$ and $x<0$, hence ${\displaystyle e^{\frac{3x+4y}{4}}>e^{x+y}}$
and $\displaystyle 1+11\cdot e^{2x+y}+4\cdot e^{x+3y}\geq$ $\displaystyle 4\cdot11^{\frac{1}{4}}\cdot e^{\frac{3x+4y}{4}}>4\cdot11^{\frac{1}{4}}\cdot e^{x+y}>6\cdot e^{x+y}$.
In particular, $\mathcal{Z}_{1}(\{4\};x,y)$ is always positive. 

At $N=3$, $n=3$ and $f_{1}\equiv x$, $f_{3}\equiv y$, $f_{4}\equiv z$
($x_{1}=x$, $x_{2}=y$, $x_{3}=z$) the stratum $\mathcal{Z}_{\mathrm{sing},1}$
contains $3$ super-ideal statistical surfaces defined by the equations
\begin{equation}
\begin{array}{c}
\mathcal{Z}_{1}(\{1\};x,y,z)\equiv-e^{x}+e^{y}+e^{z}=0,\\
\mathcal{Z}_{1}(\{2\};x,y,z)\equiv e^{x}-e^{y}+e^{z}=0,\\
\mathcal{Z}_{1}(\{3\};x,y,z)\equiv e^{x}+e^{y}-e^{z}=0
\end{array}\label{eq: components 1-stratum, ex2}
\end{equation}
and given in figure \ref{fig: super-ideal case 3D}. 
\begin{figure}[tph]
\centering{
\includegraphics[scale=0.5]{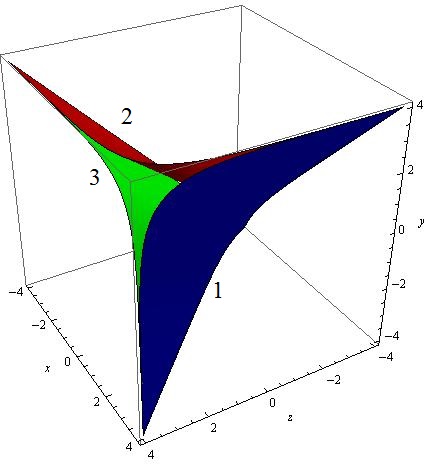}
\caption{$1$-stratum for the super-ideal case: {\small{}$f_{1}(x,y,z)\equiv x$
$f_{2}(x,y,z)\equiv y$, $f_{3}(x,y,z)\equiv z$. }}
\label{fig: super-ideal case 3D}}
\end{figure}
Induced metric $g_{ik}$, Gauss curvature $K$ and mean curvature
$\Omega$ of the surface given by the equation $\mathcal{Z}_{1}(\{1\})=0$ are
(with $y$ and $z$ as local coordinates) \cite{AK2016} 
\begin{equation}
\begin{array}{c}
g_{ik}=\delta_{ik}+w_{i}w_{k},\quad i,k=1,2,\\
K=0,\\
\Omega={\displaystyle \frac{1-T_{3}}{(1+T_{2})^{\frac{3}{2}}}}
\end{array}\label{eq: superideal 3d geometric characteristics}
\end{equation}
where probability ${\displaystyle w_{1}=\frac{e^{y}}{e^{y}+e^{z}}}$,
${\displaystyle w_{2}=\frac{e^{z}}{e^{y}+e^{z}}}$ and ${\displaystyle T_{l}=w_{1}^{l}+w_{2}^{l}}$,
$l=2,3$. For surface given by the equation $\mathcal{Z}_{1}(\{2\})=0$ and $\mathcal{Z}_{1}(\{3\})=0$
one has similar results with substitution $x\leftrightarrows y$ and
$x\leftrightarrows z$, respectively. At large $x$, $y$, $z$ surfaces
$1$ and $2$ tend to the half-plane $x=y$, $z<0$, surfaces $1$
and $3$ tend to the half-plane $x=z$, $y<0$ and surfaces $2$ and
$3$ tend to the half-plane $y=z$, $x<0$. 

Last example is
presented in figure \ref{fig: example 3D} and corresponds to $N=6$ and $n=3$ with $f_{1}(x,y,z)=0$,
$f_{2}(x,y,z)=3x$, $f_{3}(x,y,z)=3y$, $\,f_{4}(x,y,z)=2x+z+\log6$,
$f_{5}(x,y,z)=2x+y+z+\log11$, $f_{6}(x,y,z)=3y+z+\log4$. 
\begin{figure}[tph]
\centering{
\includegraphics[scale=0.5]{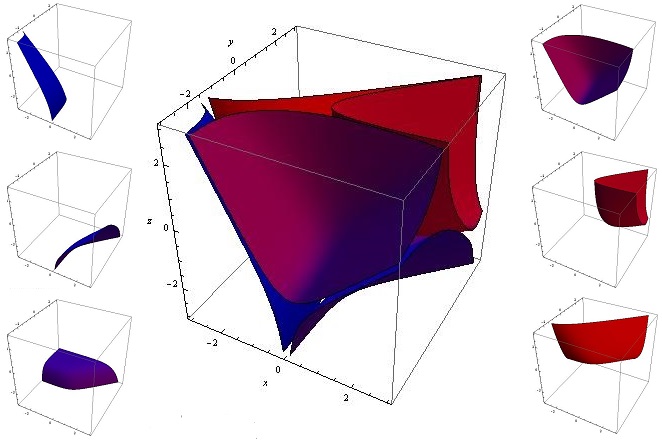}
\caption{$1$-stratum with {\small{}$f_{1}\equiv0$, $f_{2}\equiv3x$, $f_{3}\equiv3y$,
$\,f_{4}\equiv2x+z+\log6$, $f_{5}\equiv2x+y+z+\log11$, $f_{6}\equiv3y+z+\log4$}. }
\label{fig: example 3D}}
\end{figure}

\section{Higher strata \label{sec: Higher strata: zero loci}}

Higher strata $\mathcal{Z}_{\mathrm{sing},k}$ have rather complicated
structure. For instance, for $N=6,$ $n=2$ and functions $f_{\alpha}(x,y)$
given, as in figure \ref{fig: intermediate number curves 1-stratum},
by $f_{1}(x,y)\equiv0,$ $f_{2}(x,y)\equiv3x$, $f_{3}(x,y)\equiv3y$,
$f_{4}(x,y)\equiv x+y+\ln6$, $f_{5}(x,y)\equiv2x+y+\ln11$, $f_{6}(x,y)\equiv x+3y+\ln4$,
the stratum $\mathcal{Z}_{\mathrm{sing},2}$ is the set of $15$ curves
defined by ${\displaystyle \frac{N(N-1)}{2}=15}$ equations 
\begin{equation}
\mathcal{Z}_{2}(\{\alpha,\beta\};x,y):=\sum_{\gamma=1}^{6}g_{\{\alpha\beta\}(\gamma)}e^{f_{\gamma}(x,y)},\quad\alpha\neq\beta,\,\alpha,\beta=1,\dots,6\label{eq: 2-stratum N=6}
\end{equation}
with $g_{\{\alpha\beta\}(\gamma)}=-\delta_{\alpha\gamma}-\delta_{\beta\gamma}$ at $\gamma\in\{\alpha,\beta\}$,
$g_{\{\alpha\beta\}(\gamma)}=1$ at $\gamma\neq\alpha,\beta$. The stratum
$\mathcal{Z}_{\mathrm{sing},3}$ is the set of $20$ curves defined
by $\binom{N}{k}=20$ equations 
\begin{equation}
\mathcal{Z}_{3}(\{\alpha,\beta,\gamma\};x,y):=\sum_{\eta=1}^{6}g_{\{\alpha\beta\gamma\}(\eta)}e^{f_{\eta}(x,y)},\quad\alpha\neq\beta\neq\gamma\neq\alpha,\,\alpha,\beta,\gamma=1,\dots,N\label{eq: 3-stratum N=6}
\end{equation}
with $g_{\{\alpha\beta\gamma\}(\eta)}=-\delta_{\alpha\eta}-\delta_{\beta\eta}-\delta_{\gamma\eta}$ at $\eta\in\{\alpha,\beta,\gamma\}$,
$g_{\{\alpha\beta\gamma\}(\eta)}=1$ at $\eta\neq\alpha,\beta,\gamma$.
These sets of curves are presented in figure \ref{fig: examples of 1-strata, a}
and \ref{fig: examples of 1-strata, b} respectively. 
\begin{figure}[tph]
\centering
\subfigure[$2$-stratum. ]{\includegraphics[scale=0.35]{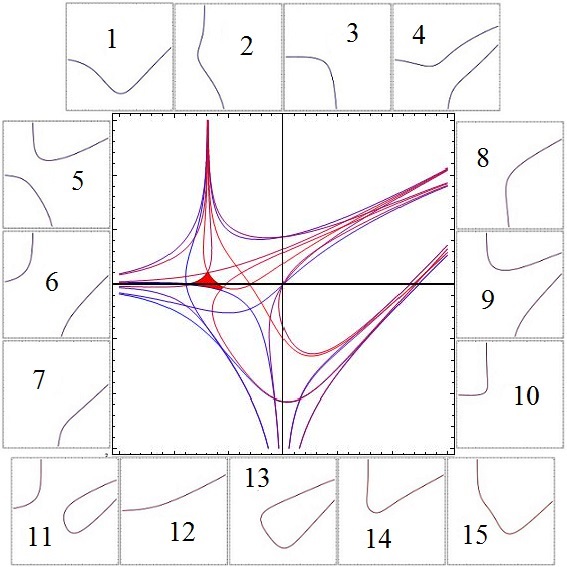}\label{example higher strata, 2}}\hfill{}\subfigure[$3$-stratum. ]{\includegraphics[scale=0.35]{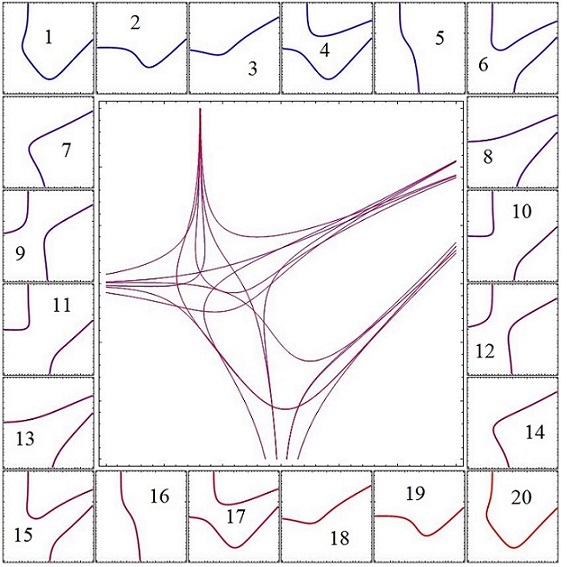}\label{example higher strata, 3}}\caption{Higher strata with $f_{1}\equiv 0$, $f_{2}\equiv 3x$, $f_{3}\equiv 3y$,
$f_{4}\equiv x+y+\ln 6$, $f_{5}\equiv 2x+y+\ln 11$, $f_{6}\equiv x+3y+\ln 4$.
For each choice of $2$- and $3$- subsets, listed in lexicographical
order, corresponding locus is shown. Components in $\mathcal{Z}_{\mathrm{sing},3}$ are listed twice: curve $\alpha$ coincides with curve $21-\alpha$, $\alpha=1,\dots,10$, since they correspond to the same partition. \label{fig: example higher strata}}
\end{figure}

Due to the involution $P:\,g_{\alpha}\mapsto-g_{\alpha}$ the strata
$\mathcal{Z}_{\mathrm{sing},4}$ and $\mathcal{Z}_{\mathrm{sing},5}$
coincide with strata $\mathcal{Z}_{\mathrm{sing},2}$ and $\mathcal{Z}_{\mathrm{sing},1}$,
respectively. The stratum $\mathcal{Z}_{\mathrm{sing},3}$ is stable
under involution, $P(\mathcal{Z}_{\mathrm{sing},3})=\mathcal{Z}_{\mathrm{sing},3}$
and at most half of conditions in (\ref{eq: 3-stratum N=6})
are independent. 

In the case $N=6$ and $n=3$ and with the same functions $f_{\alpha}$
as in figure \ref{fig: example 3D}, i.e. $f_{1}\equiv0$, $f_{2}\equiv3x$,
$f_{3}\equiv3y$, $\,f_{4}\equiv2x+z+\log6$, $f_{5}\equiv2x+y+z+\log11$,
$f_{6}\equiv3y+z+\log4$, the strata $\mathcal{Z}_{\mathrm{sing},2}$
and $\mathcal{Z}_{\mathrm{sing},3}$ are given in figure \ref{fig: examples higher strata 3D}.
\begin{figure}[tph]
\centering{
\subfigure[$2$-stratum. ]
{\includegraphics[scale=0.4]{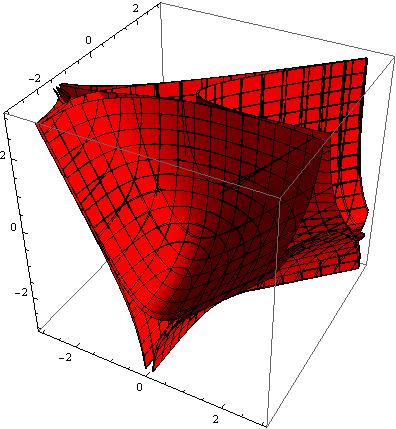}
\label{fig: example 2-strata 3D}}\hfill{}
\subfigure[$3$-stratum. ]
{\includegraphics[scale=0.4]{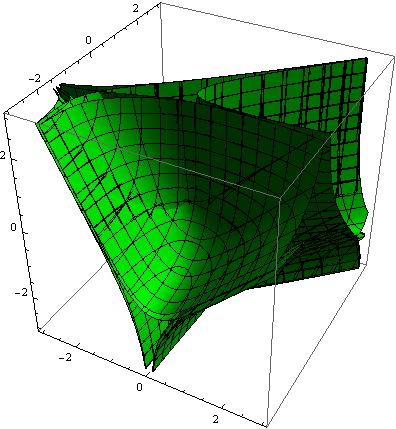}
\label{fig: example 3-strata 3D}}
\caption{Examples of higher strata with {\small{}$f_{1}\equiv0$, $f_{2}\equiv3x$,
$f_{3}\equiv3y$, $\,f_{4}\equiv2x+z+\log6$, $f_{5}\equiv2x+y+z+\log11$,
$f_{6}\equiv3y+z+\log4$}. }
\label{fig: examples higher strata 3D}
}
\end{figure}

In order to describe general properties of higher singular strata
let us introduce some notation. We will denote by $(\mathcal{I}_{1},\mathcal{I}_{2})$
an ordered $2$-partition of the set $[N]:=\{1,\dots,N\}$, i.e.
the pair of two subsets $\mathcal{I}_{1},\mathcal{I}_{2}\subseteq[N]$
such that $\mathcal{I}_{1}\cup\mathcal{I}_{2}=[N]$ and $\mathcal{I}_{1}\cap\mathcal{I}_{2}=\emptyset$.
Then for each ordered $2$-partition $(\mathcal{I},[N]\backslash\mathcal{I})$
we define the function 
\begin{equation}
\mathcal{Z}_{k}(\mathcal{I};\boldsymbol{x}):=-\sum_{\alpha\in\mathcal{I}_{1}}e^{f_{\alpha}(\boldsymbol{x})}+\sum_{\beta\in\mathcal{I}_{2}}e^{f_{\beta}(\boldsymbol{x})}\label{eq: signed partition}
\end{equation}
and the corresponding zero locus as 
\begin{equation}
\mathcal{Z}_{\mathrm{sing},k}(\mathcal{I}):=\{\boldsymbol{x}:\mathcal{Z}_{k}(\mathcal{I};\boldsymbol{x})=0\}.\label{eq: locus signed partition}
\end{equation}
In the following we will ofter write $\mathcal{Z}_{k}(\mathcal{I})$ instead of $\mathcal{Z}_{k}(\mathcal{I};\boldsymbol{x})$ for notational convenience.

Since $\mathcal{Z}_{N-k}(\mathcal{I}_{2},\mathcal{I}_{1})=-\mathcal{Z}_{k}(\mathcal{I}_{1},\mathcal{I}_{2})$
zero loci of $\mathcal{Z}_{k}(\mathcal{I}_{1},\mathcal{I}_{2})$ and $\mathcal{Z}_{N-k}(\mathcal{I}_{2},\mathcal{I}_{1})$
coincide. To avoid such redundancy we will focus on unordered partitions
$\{\mathcal{I}_{1},\mathcal{I}_{2}\}$ and we will assume in what
follows for each partition $(\mathcal{I}_{1},\mathcal{I}_{2})$ the
cardinality $\#\mathcal{I}_{1}$ of the first subset $\mathcal{I}_{1}$
is smaller than that of $\mathcal{I}_{2}$. Since $\mathcal{I}_{1}\cup\mathcal{I}_{2}=[N]$
and $\mathcal{I}_{1}\cap\mathcal{I}_{2}=\emptyset$ one has ${\displaystyle 0\leq\#\mathcal{I}_{1}\leq \left\lfloor\frac{N}{2}\right\rfloor}$. 

Further for each subset $\mathcal{I}$ of $[N]$ of cardinality
$k$ one has the equation 
\begin{equation}
\mathcal{Z}_{k}(\mathcal{I};\boldsymbol{x})=0\label{eq: locus signed partition with assumption}
\end{equation}
which defines the hypersurface (\ref{eq: locus signed partition}) in $\mathbb{R}^{n}$. Union of all such
hypersurfaces with $\#\mathcal{I}=k$ is the stratum $\mathcal{Z}_{\mathrm{sing},k}$.
Denoting the set of all subsets $\mathcal{I}$ of $[N]$ with $\#\mathcal{I}=k$
as $\mathcal{P}_{k}[N]$ one, hence, has 
\begin{equation}
\mathcal{Z}_{\mathrm{sing},k}=\bigcup_{\mathcal{I}\in\mathcal{P}_{k}[N]}\mathcal{Z}_{\mathrm{sing},k}(\mathcal{I}) .\label{eq: k-stratum}
\end{equation}
Note also that
\begin{equation}
{\displaystyle \mathcal{Z}_{k}(\mathcal{I};\boldsymbol{x})=0\Leftrightarrow2\cdot\sum_{\alpha\in\mathcal{I}}e^{f_{\mathcal{\alpha}}(\boldsymbol{x})}=\mathcal{Z}_{0}(\boldsymbol{x})}\label{eq: subsets and elements partition functions relation-1}
\end{equation}
and 
\begin{equation}
{\displaystyle \mathcal{Z}_{k}(\boldsymbol{x}):=\sum_{\mathcal{I}\in\mathcal{P}_{k}[N]}e^{f_{\mathcal{I}}(\boldsymbol{x})}=\binom{N-1}{k-1}\cdot\mathcal{Z}_{0}(\boldsymbol{x})}\label{eq: subsets and elements partition functions relation}
\end{equation}
where ${\displaystyle \mathcal{Z}_{0}(\boldsymbol{x})=\sum_{\alpha=1}^{N}e^{f_{\alpha}(\boldsymbol{x})}}$
and we denote 
\begin{equation}
e^{f_{\mathcal{I}}(\boldsymbol{x})}:=\sum_{\alpha\in\mathcal{I}}e^{f_{\alpha}(\boldsymbol{x})}.\label{eq: renormalized function subsets}
\end{equation}

Figure \ref{fig: comments higher strata} indicates that curves and
hypersurfaces which compose higher strata may intersect in contrast
to the stratum $\mathcal{Z}_{\mathrm{sing},1}$. 
\begin{figure}[tph]
\centering
\subfigure[Crossing of hypersurfaces can happen for a higher stratum $\mathcal{Z}_{\mathrm{sing},2}$. Not all hypersurfaces lie in the same halfspace defined by a certain hypersurface, as it is for the orange region defined by $\mathcal{Z}_{2}(\{1, 2\})>0$. Red region is the set where $\mathcal{Z}_{2}(\mathcal{I})>0$ for all $\mathcal{I}$ with $\#\mathcal{I}=2$. ]
{\includegraphics[scale=0.30]{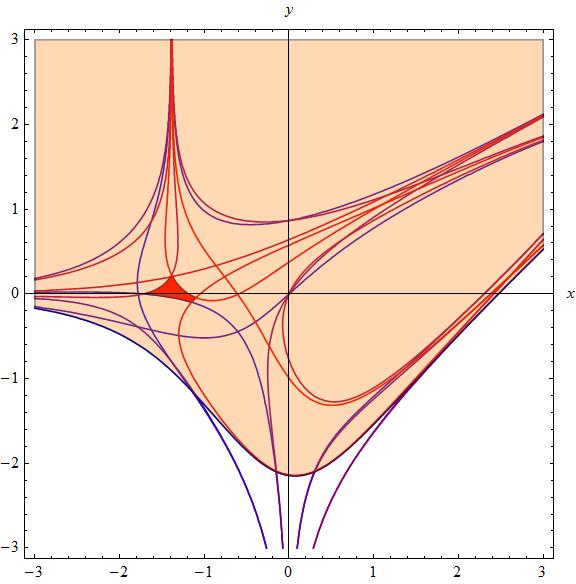}\label{fig: comments higher strata, a}}\hfill{}\subfigure[Hypersurface defined by $\mathcal{Z}_{2}(\{1,2\};\boldsymbol{x})=0$.
The red line shows that both regions of the plane bounded by this
component are non-convex. ]
{\includegraphics[scale=0.42]{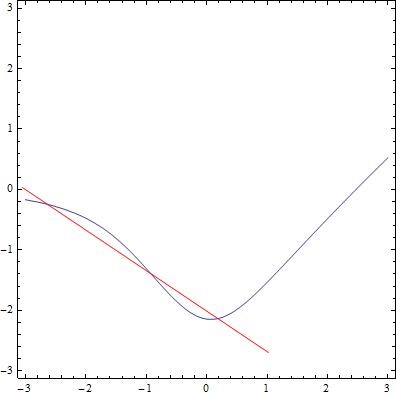}\label{fig: comments higher strata, b}}
\caption{Comments on higher stratum $\mathcal{Z}_{\mathrm{sing},2}$ with $f_{1}\equiv0$,
$f_{2}\equiv3x$, $f_{3}\equiv3y$, $f_{4}\equiv x+y+\ln6$, $f_{5}\equiv2x+y+\ln11$,
$f_{6}\equiv x+3y+\ln4$. }
\label{fig: comments higher strata}
\end{figure}
In general, one has 
\begin{prop}
\label{prop: common solutions intersection} Two hypersurfaces (\ref{eq: locus signed partition with assumption})
belonging to the same stratum $\mathcal{Z}_{\mathrm{sing},k}$ and
different $\mathcal{I}_{1}$ and $\mathcal{I}_{2}$ intersect at finite
$\boldsymbol{x}$ only if $\mathcal{I}_{1}\cap\mathcal{I}_{2}\neq\emptyset$. \end{prop}
\begin{proof}
Consider two hypersurfaces $\mathcal{Z}_{\mathrm{sing},k}(\mathcal{I}_{1})$
and $\mathcal{Z}_{\mathrm{sing},k}(\mathcal{I}_{2})$ associated
with two partitions of $[N]$ with $\#\mathcal{I}_{1}=\#\mathcal{I}_{2}=k$.
Assume that they have a common point $\boldsymbol{\tilde{x}}$. Hence,
one has \\
\[
\mathcal{Z}_{k}(\mathcal{I}_{1};\boldsymbol{\tilde{x}})+\mathcal{Z}_{k}(\mathcal{I}_{2};\boldsymbol{\tilde{x}})=2\cdot\sum_{\alpha\in[N]\backslash(\mathcal{I}_{1}\cup\mathcal{I}_{2})}e^{f_{\alpha}(\boldsymbol{\tilde{x}})}-2\cdot\sum_{\beta\in\mathcal{I}_{1}\cap\mathcal{I}_{2}}e^{f_{\beta}(\boldsymbol{\tilde{x}})}=0.
\]
If $\mathcal{I}_{1}\cup\mathcal{I}_{2}\neq[N]$ then this equation
may have real solutions only if the second term is different from
zero, i.e. $\mathcal{I}_{1}\cap\mathcal{I}_{2}\neq\emptyset$. If
$\mathcal{I}_{1}\cup\mathcal{I}_{2}=[N]$ then ${\displaystyle 2k\geq N}$,
but ${\displaystyle k\leq\frac{N}{2}}$ so ${\displaystyle k=\frac{N}{2}}$;
then $\mathcal{I}_{1}$ and $\mathcal{I}_{2}$ are complementary sets
and they define the same equation $\mathcal{Z}_{k}(\mathcal{I}_{1},\mathcal{I}_{2})=0$. 
\end{proof}
Maximum number of intersections in stratum $\mathcal{Z}_{\mathrm{sing},k}$
is bounded by one-half of the number of different pairs $\mathcal{I}_{1}$
and $\mathcal{I}_{2}$ with ${\displaystyle \#\mathcal{I}_{1}=\#\mathcal{I}_{2}=k\leq\left\lfloor \frac{N}{2}\right\rfloor }$
such that $\mathcal{I}_{1}\cap\mathcal{I}_{2}\neq\emptyset$. There
are ${\displaystyle \binom{N}{k}}$ different partitions $(\mathcal{I},[N]\backslash\mathcal{I})$
of $[N]$ with cardinality $\#\mathcal{I}=k$ and for each such $\mathcal{I}$
there are ${\displaystyle \binom{N}{k}-1-\binom{N-k}{k}}$ partitions $(\mathcal{J},[N]\backslash\mathcal{J})$,
$\mathcal{J}\neq\mathcal{I}$ and $\mathcal{I}\cap\mathcal{J}\neq\emptyset$.
Then the number of intersections in the stratum $\mathcal{Z}_{\mathrm{sing},k}$
is bounded from above by ${\displaystyle \frac{1}{2}\cdot \binom{N}{k}\cdot\left[\binom{N}{k}-1-\binom{N-k}{k}\right]}$.
It is not always a strict bound since some of these intersections
might be unreachable. For example, at $k=2$ one has $\#(\mathcal{I}_{1}\cap\mathcal{I}_{2})=1$
and, hence, $\mathcal{I}_{1}=\{\alpha,\gamma_{1}\}$ and $\mathcal{I}_{2}=\{\alpha,\gamma_{2}\}$
for some $\alpha\neq\gamma_{1}\neq\gamma_{2}\neq\alpha$. Then, if
$\boldsymbol{\bar{x}}$ is in the intersection $\mathcal{Z}_{\mathrm{sing},2}(\mathcal{I}_{1})\cap\mathcal{Z}_{\mathrm{sing},2}(\mathcal{I}_{2})$
one has ${\displaystyle e^{f_{\gamma_{1}}(\bar{\boldsymbol{x}})}-e^{f_{\gamma_{2}}(\bar{\boldsymbol{x}})}=-e^{f_{\alpha}(\boldsymbol{\bar{x}})}+\sum_{\beta\neq\alpha,\gamma_{1},\gamma_{2}}e^{f_{\beta}(\boldsymbol{\bar{x}})}=e^{f_{\gamma_{2}}(\boldsymbol{\bar{x}})}-e^{f_{\gamma_{1}}(\boldsymbol{\bar{x}})}}$,
thus $f_{\gamma_{1}}(\boldsymbol{\bar{x}})=f_{\gamma_{2}}(\boldsymbol{\bar{x}})$
and ${\displaystyle e^{f_{\alpha}(\boldsymbol{\bar{x}})}=\sum_{\beta\neq\alpha,\gamma_{1},\gamma_{2}}e^{f_{\beta}(\boldsymbol{\bar{x}})}}$.
This last equation not always has real solutions for general choice
of functions. 

In general, intersections of hypersurfaces (\ref{eq: locus signed partition with assumption})
for the stratum $\mathcal{Z}_{\mathrm{sing},k}$ are not necessarily
transversal. Moreover, the assumption of real-valued functions opens
up the way to reductions. For example,
take two functions $g(\boldsymbol{x})$ and
$h(\boldsymbol{x})$ and consider $f_{1}(\boldsymbol{x})\equiv 4\cdot g(\boldsymbol{x})$,
$f_{2}(\boldsymbol{x})\equiv 4\cdot h(\boldsymbol{x})$, $f_{3}(\boldsymbol{x})\equiv\ln 6+2\cdot g(\boldsymbol{x})+2\cdot h(\boldsymbol{x})$,
$f_{4}(\boldsymbol{x})=\ln 4+3\cdot g(\boldsymbol{x})+h(\boldsymbol{x})$,
$f_{5}(\boldsymbol{x})=\ln 4+g(\boldsymbol{x})+3\cdot h(\boldsymbol{x})$.
Then $\mathcal{Z}_{2}(\{4,5\};\boldsymbol{x})={\displaystyle \left(e^{2g(\boldsymbol{x})}+e^{2h(\boldsymbol{x})}-e^{\ln2+g(\boldsymbol{x})+h(\boldsymbol{x})}\right)^{2}}\geq0$
and it vanishes if and only if $e^{2g(\boldsymbol{x})}+e^{2h(\boldsymbol{x})}=e^{\ln2+g(\boldsymbol{x})+h(\boldsymbol{x})}$.
From Arithmetic-Geometric inequality, this is equivalent to the algebraic
constraint $g(\boldsymbol{x})=h(\boldsymbol{x})$. Occurrence of such
non-transversal crossings or reductions is a particular case and influences
the investigation on equilibrium and non-equilibrium regions. We assume hereafter that pairwise intersections
between hypersurfaces defined by $ \mathcal{Z}_{k}(\mathcal{I})=0 $ or $f_{\alpha}(\boldsymbol{x})-f_{\beta}(\boldsymbol{x})=0$, $1\leq\alpha<\beta\leq N$
are transversal.

\section{Higher strata. General properties \label{sec: Higher strata: general properties}}

All hypersurfaces (\ref{eq: locus signed partition})
which compose the stratum $\mathcal{Z}_{\mathrm{sing},k}$ divide
$\mathbb{R}^{n}$ in a number of regions which we will call domains.
Then, let us denote the subdomain in $\mathbb{R}^{n}$ where all functions
$\mathcal{Z}_{k}(\mathcal{I})>0$, $\mathcal{I}\in\mathcal{P}_{k}[N]$, as $\mathcal{D}_{k+}$. For example, blue regions in figures \ref{fig: examples of 1-strata},\ref{fig: examples number curves} and \ref{fig: intermediate number curves 1-stratum} represent $\mathcal{D}_{1+}$ and the red region in \ref{fig: comments higher strata, a} represents $\mathcal{D}_{2+}$. 
The domain $\mathbb{R}^{n}\backslash\mathcal{D}_{k+}$ is divided
by hypersurfaces $\left\{ \mathcal{Z}_{k}(\mathcal{I})=0\right\} $ into subdomains
where some of functions $\mathcal{Z}_{k}(\mathcal{I})$ are positive and others
are negative. Let us denote $\mathcal{D}_{k-}$ as the domain where
the number of negative functions $\mathcal{Z}_{k}(\mathcal{J})$ with $\mathcal{J}\in\mathcal{P}_{k}[N]$
is maximal. The hypersurfaces of the $k$-stratum are confined in
certain domain which we will refer as zeros confinement domain $ZCD_{k}:=\mathbb{R}^{n}\backslash(\mathcal{D}_{k+}\cup\mathcal{D}_{k-})$.
The domain $ZCD_{k}$ is a sort of intermittent shell which separate
the domains $\mathcal{D}_{k+}$ and $\mathcal{D}_{k-}$ and the boundary
of $\mathcal{D}_{k+}\cup ZCD_{k}$ will be referred as the extremal
points of hypersurfaces composing $\mathcal{Z}_{\mathrm{sing},k}$.
For the first stratum $V_{1}$ defined in (\ref{eq: stratification variable space})
the domain $ZCD_{1}$ generically has dimension $n-1$ and consists
of hypersurfaces $\left\{\boldsymbol{x}: \mathcal{Z}_{1}(\{\alpha\};\boldsymbol{x})=0\right\} $, $\alpha=1,\dots,N$
themselves. For higher $V_{k}$, $k\geq2$, the domain $ZCD_{k}$
has generically dimension $n$ and its boundary is tipically formed
by pieces of different hypersurfaces belonging to $\mathcal{Z}_{\mathrm{sing},k}$. 

Since at $\mathcal{D}_{k+}$ the partition function $\mathcal{Z}_{k}\equiv\mathcal{Z}|_{V_{k}}>0$
it is natural to refer to the domain $\mathcal{D}_{k+}$ as stability
(equilibrium) domain. It is surrounded by the domain $ZCD_{k}$ with
rather complicated singularity structure. The domain $\mathcal{D}_{k-}$
is an ambient instability domain.  

These domains for different strata exhibit a simple inclusion chain. 
\begin{prop}
\label{prop: inclusion chain equilibrium domain} Let ${\displaystyle 1\leq k<\hat{k}\leq\left\lfloor \frac{N}{2}\right\rfloor }$
then $\mathcal{D}_{\hat{k}+}\subseteq\mathcal{D}_{k+}$. \end{prop}
\begin{proof}
Through the proof is an immediate consequence of the definition of $\mathcal{D}_{k+}$, we present it here for completeness. Let $\hat{k}>k$ and $\mathcal{I}_{k}\subset\mathcal{I}_{\hat{k}}$
be two subsets of $[N]$ such that $\#\mathcal{I}_{k}=k$ and $\#\mathcal{I}_{\hat{k}}=\hat{k}$.
Then for any $\boldsymbol{x}\in\mathbb{R}^{n}$ one has the identity
\begin{equation}
\mathcal{Z}_{k}(\mathcal{I}_{k};\boldsymbol{x})-\mathcal{Z}_{\hat{k}}(\mathcal{I}_{\hat{k}};\boldsymbol{x})=2\cdot \sum_{\alpha\in\mathcal{I}_{k+1}\backslash\mathcal{I}_{k}} e^{f_{\alpha}(\boldsymbol{x})}>0.\label{eq: identity monotony signed partitions}
\end{equation}
In the domain $\mathcal{D}_{\hat{k}+}$ one has $\mathcal{Z}_{\hat{k}}(\mathcal{I}_{\hat{k}};\boldsymbol{x})>0$
for all subsets $\mathcal{I}_{k+1}$. Inequality (\ref{eq: identity monotony signed partitions})
implies that all $\mathcal{Z}_{k}(\mathcal{I}_{k};\boldsymbol{x})>0$ in $\mathcal{D}_{\hat{k}+}$
too. So $\mathcal{D}_{\hat{k}+}\subseteq\mathcal{D}_{k+}$. 
\end{proof}
Thus, one has the inclusion chain 
\begin{equation}
\mathcal{D}_{0+}=\mathbb{R}^{n}\supseteq\mathcal{D}_{1+}\supseteq\mathcal{D}_{2+}\supseteq\dots\supseteq\mathcal{D}_{\left\lfloor \frac{N}{2}\right\rfloor +}.\label{eq: inclusion chain equilibrium domain}
\end{equation}
An example with $N=7$ and $f_{\alpha}(x,y)$ as in (\ref{eq: example family maximal number 1-components}) is given in figure \ref{fig: chain equilibrium}, where the domain  $\mathcal{D}_{1+}$ is shown in blue color, $\mathcal{D}_{2+}$ in red and $\mathcal{D}_{3+}$ in green. 
\begin{figure}[tph]
\centering
\subfigure[Singular sectors $\mathcal{Z}_{\mathrm{sing},1}$ (blue), $\mathcal{Z}_{\mathrm{sing},2}$ (red) and $\mathcal{Z}_{\mathrm{sing},3}$ (green). ]{
\includegraphics[scale=0.4]{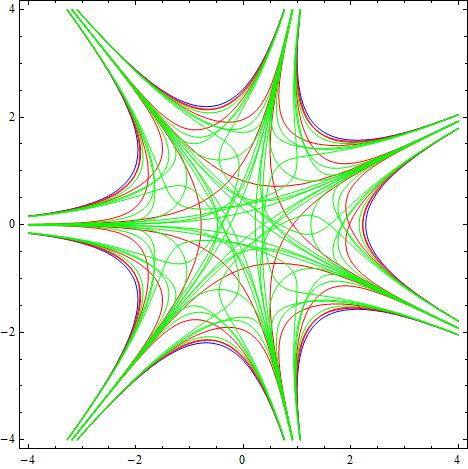}
\label{fig: chain equilibrium, a}}
\hfill{}
\subfigure[Domains $\mathcal{D}_{3+}\subseteq\mathcal{D}_{2+}\subseteq{D}_{1+}$. ]
{
\includegraphics[scale=0.4]{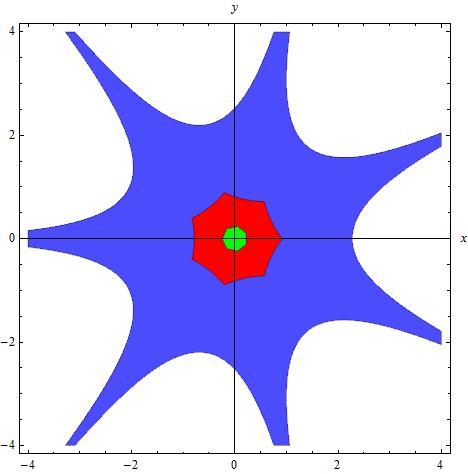}
\label{fig: chain equilibrium, b}}
\caption{Inclusion chain for equilibrium domains in the case $f_{\alpha}\equiv\cos\left(\frac{2\pi(\alpha-1)}{7}\right)x+\sin\left(\frac{2\pi(\alpha-1)}{7}\right)y$, $\alpha=1,\dots,7$. }
\label{fig: chain equilibrium}
\end{figure}

Furthermore, one also has 
\begin{prop}
\label{prop: inclusion chain equilibrium + ZCD} 
Let us take $\displaystyle 1\leq k<\hat{k}\leq\left\lfloor \frac{N}{2}\right\rfloor$. Then, there exists a dense subset
of $\mathcal{Z}_{\mathrm{sing},\hat{k}}$ such that each ray $(\vec{r})_{i}=x_{0,i}+t\cdot e_{i}$ originated
from this set, $\boldsymbol{e}\in\mathbb{R}^{n}$,
intersects $\mathcal{Z}_{\mathrm{sing},k}$.
\end{prop}
In other words, if $\{f_{\alpha}(\boldsymbol{x}):\,\alpha\in[N]\}$
are pairwise different linear functions, then $\mathcal{Z}_{\mathrm{sing},\hat{k}}$
lies inside a region of $\mathbb{R}^{n}$ delimited by some components
of $\mathcal{Z}_{\mathrm{sing},k}$, ${\displaystyle 1\leq k<\hat{k}\leq\left\lfloor \frac{N}{2}\right\rfloor }$. 
\begin{proof}
The proof is based on the following 
\begin{lem}
\label{lem: ray lemma} Consider any ray $\vec{r}$ with base point
$\boldsymbol{x}$, slopes $\boldsymbol{e}:=(e_{1},\dots,e_{n})$
and parametrization $\vec{r}(t):=\boldsymbol{x}+t\cdot\boldsymbol{e}$,
$t\geq0$. If none of functions $f_{\alpha}(\vec{r}(t))-f_{\beta}(\vec{r}(t))$
vanishes identically, $\alpha\neq\beta$, then there exists $t_{0}\in\mathbb{R}_{+}$
such that $\vec{r}(t)$ belongs to a region in $\mathbb{R}^{n}$ where
there is only one dominant function $f_{\alpha}$, $\alpha\in[N]$, at $t\geq t_{0}$, i.e. $\#\left\{ \alpha\in[N]:\,\forall\beta\in[N],\,f_{\alpha}(\vec{r}(t))\geq f_{\beta}(\vec{r}(t))\right\} =1$. \end{lem}
\begin{proof}
Let us consider the ${\displaystyle \frac{N(N-1)}{2}}$ functions
\begin{equation}
d_{\alpha\beta}(t):= f_{\alpha}(\vec{r}(t))-f_{\beta}(\vec{r}(t))\quad1\leq\alpha<\beta\leq N.\label{eq: difference phases}
\end{equation}
Since none of these linear functions vanishes identically, each of them
has a finite number of roots, so the set of points 
\begin{equation}
\Omega:=\displaystyle \bigcup_{1\leq\alpha<\beta\leq N}\{t:\,d_{\alpha\beta}(t)=0\}
\label{eq: set of roots lemma}
\end{equation}
is finite. Thus, for $t_0>\max(\Omega)$, all $d_{\alpha\beta}(t_0)$
will be definitely different from zero, hence all $\{f_{\alpha}(\vec{r}(t_0)):\,\alpha\in[N]\}$
are pairwise different. In particular, there will be one and only one $\alpha_{0}\in[N]$
such that $f_{\alpha_{0}}(\vec{r}(t_0))>f_{\alpha}(\vec{r}(t_0))$ for
all $\alpha\neq \alpha_{0}$. So $f_{\alpha_{0}}(\vec{r}(t))-f_{\alpha}(\vec{r}(t))>0$ at $t>t_0$ too, since a change of sign would imply an additional zero
of $d_{\alpha_{0}\alpha}$ by continuity. 
\end{proof}

Now, let ${\displaystyle (\mathcal{J},}[N]\backslash\mathcal{J})$ be any
partition of $[N]$ such that ${\displaystyle \#\mathcal{J}=\hat{k}\leq\left\lfloor \frac{N}{2}\right\rfloor }$
and $\vec{r}(t):=\left(x_{i}+t\cdot e_{i}\right)$, $t\geq0$
be a ray with base point $\boldsymbol{x}\in\mathcal{Z}_{\mathrm{sing},\hat{k}}$.
If $d_{\alpha\beta}(t)$ in (\ref{eq: difference phases}) does not
vanish identically for all $\alpha<\beta$, then lemma \ref{lem: ray lemma} implies that there exist $t_{0}>0$ and $\alpha_{0}\in[N]$ such that one has
$d_{\beta\alpha_{0}}(t)<0$ at $t\geq t_{0}$ and $\beta\neq \alpha_{0}$.
In particular, $d_{\beta \alpha_{0}}$ are linear functions of
$t$ and they are negative at $t\geq t_{0}$. Hence one has 
\begin{equation}
{\displaystyle \lim_{t\rightarrow+\infty}f_{\beta}(\vec{r}(t))-f_{\alpha_{0}}(\vec{r}(t))=-\infty},\quad\beta\neq \alpha_{0}.\label{eq: prop limit difference phases}
\end{equation}
So there exists $t_{1}>t_{0}$ such that ${\displaystyle \sum_{\beta\neq\alpha_{0}}e^{f_{\beta}(\vec{r}(t))-f_{\alpha_{0}}(\vec{r}(t))}<1}$
at $t>t_{1}$. The index $\alpha_{0}$ belongs to only one subset
$\mathcal{J}$ or $[N]\backslash\mathcal{\mathcal{J}}$, call it $\mathcal{J}(\alpha_{0})$:
from ${\displaystyle k<\hat{k}\leq N-\hat{k}}$ it follows that $k<\#\mathcal{J}(\alpha_{0})$
and one can always choose a subset $\mathcal{I}\subset\mathcal{J}(\alpha_{0})$
with $k$ elements such that $\alpha_{0}\in\mathcal{I}$. One has
$\vec{r}(t)|_{t=0}=\boldsymbol{x}$ so 
 $\displaystyle \sum_{\alpha\in\mathcal{I}}e^{f_{\alpha}(\boldsymbol{x})}<\sum_{\alpha\in\mathcal{J}(\alpha_{0})}e^{f_{\alpha}(\boldsymbol{x})}
=$ $\displaystyle\sum_{\beta\in[N]\backslash\mathcal{J}(\alpha_{0})}e^{f_{\beta}(\boldsymbol{x})}<\sum_{\beta\in[N]\backslash\mathcal{I}}e^{f_{\beta}(\boldsymbol{x})}$.
On the other hand, at $t=t_{1}$ one has 
 ${\displaystyle \sum_{\alpha\in\mathcal{I}}e^{f_{\alpha}(\vec{r}(t_{1}))}>e^{f_{\alpha_{0}}(\vec{r}(t_{1}))}>\sum_{\beta\neq \alpha_{0}}e^{f_{\beta}(\vec{r}(t_{1}))}>\sum_{\beta\in[N]\backslash\mathcal{I}}e^{f_{\beta}(\vec{r}(t_{1}))}}$.
Then, there exists a point $0<\bar{t}<t_{1}$ such that ${\displaystyle \sum_{\alpha\in\mathcal{I}}e^{f_{\alpha}(\vec{r}(\bar{t}))}=\sum_{\beta\in[N]\backslash\mathcal{I}}e^{f_{\beta}(\vec{r}(\bar{t}))}}$
by continuity. Thus $\vec{r}(\bar{t})$ is in the set $ \mathcal{Z}_{\mathrm{sing},k}(\mathcal{I}) \subseteq\mathcal{Z}_{\mathrm{sing},k}$.
If instead $d_{\alpha_{0}\beta_{0}}(t)\equiv 0$ for some $\alpha_{0}\neq\beta_{0}$
at the point $(\boldsymbol{x_{0}};\boldsymbol{e})$, then $f_{\alpha_{0}}(\boldsymbol{x_{0}})-f_{\beta_{0}}(\boldsymbol{x_{0}})=d_{\alpha_{0}\beta_{0}}(0)=0$
and $\boldsymbol{x_{0}}\in\mathcal{Z}_{\mathrm{sing},k}\cap\left\{ \boldsymbol{x}\in\mathbb{R}^{n}:\,f_{\alpha_{0}}(\boldsymbol{x})=f_{\beta_{0}}(\boldsymbol{x})\right\} $.
From the generic hypothesis of transversal crossing, the complement
of set of such points $\boldsymbol{x_{0}}$ in $\mathcal{Z}_{\mathrm{sing},\hat{k}}$
is dense. 
\end{proof}
Some concrete examples are presented in figure \ref{fig: rays bounding examples}. 

\begin{figure}[tph]
\centering{
\subfigure[Stratification of $\mathcal{Z}_{\mathrm{sing},3}$ (magenta), $\mathcal{Z}_{\mathrm{sing},2}$
(blue) and $\mathcal{Z}_{\mathrm{sing},1}$ (orange) in case $f_{1}\equiv0$,
$f_{2}\equiv3x$, $f_{3}\equiv3y$, $f_{4}\equiv x+y+\ln6$, $f_{5}\equiv2x+y+\ln11$,
$f_{6}\equiv x+3y+\ln4$. ]
{\includegraphics[scale=0.42]{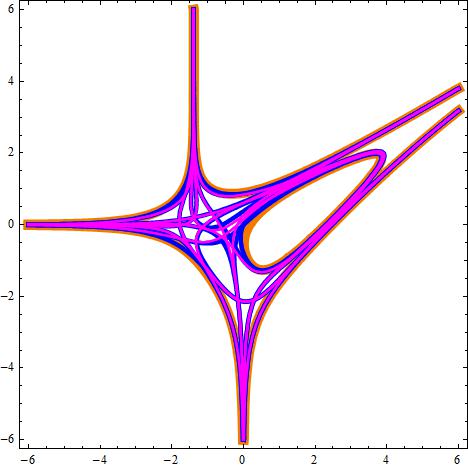}}
\hfill{}
\subfigure[Stratification of $\mathcal{Z}_{\mathrm{sing},3}$ (green) and $\mathcal{Z}_{\mathrm{sing},2}$
(red) in $N=3$ super-ideal case. ]
{\includegraphics[scale=0.42]{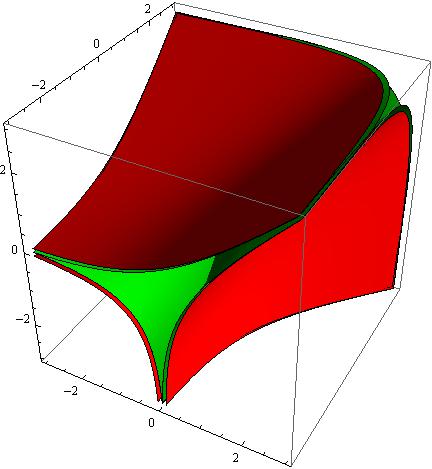}}
\caption{Examples of stratifications. }
\label{fig: rays bounding examples}
}
\end{figure}

Note that the rays that do not
come out $ZCD_{k}$ are connected with locus of coincident dominant
functions $f_{\alpha}(\boldsymbol{x})=f_{\beta}(\boldsymbol{x})$ $\displaystyle=\max_{\gamma}\left\{ f_{\gamma}(\boldsymbol{x})\right\}$,
$\alpha\neq\beta$, see figure \ref{fig: extremelocirays}.

\begin{figure}[tph]
\centering{
\subfigure[$2$-stratum $\mathcal{Z}_{\mathrm{sing},2}$; thick locus is the
set of extremal points. ]
{\includegraphics[scale=0.42]{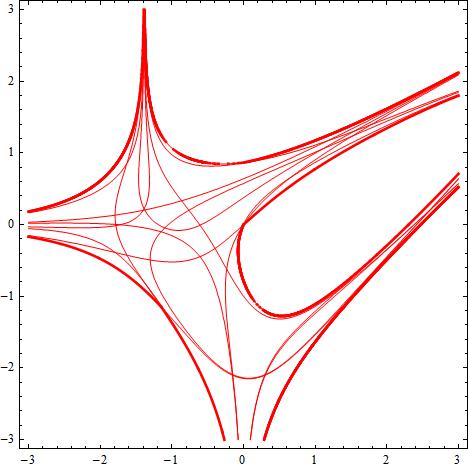}}
\hfill{}
\subfigure[Ray from $1$, an extremal point of $\mathcal{Z}_{\mathrm{sing},2}$, meets $\mathcal{Z}_{\mathrm{sing},1}$ at point $2$. The ray that comes out from point $2$ meets other strata, but the last one is $\mathcal{Z}_{\mathrm{sing},1}$ itself. Ray from point $3\in\mathcal{Z}_{\mathrm{sing},1}$ does not meet other strata. Point $4$ is in a particular position and never comes out $ZCD_{2}$. ]
{\includegraphics[scale=0.42]{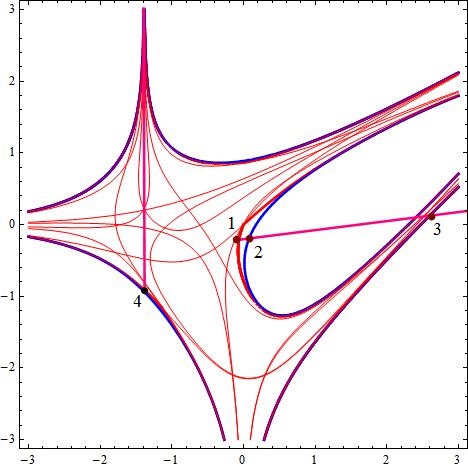}\label{fig: extremelocirays}}
\caption{Comments on relative positions of singular loci. }
\label{fig: rays bounding comments}
}
\end{figure}

\section{Statistical amoebas vs. algebraic amoebas \label{sec: Statistical amoebas vs. algebraic amoebas}}

We will refer to the domains $\mathcal{A}_{k}:=\mathcal{D}_{k+}\cup ZCD_{k}$
described in previous section as the statistical $k$-amoebas. They,
generically, are composed by the internal stable nuclei (domains $\mathcal{D}_{k+}$)
and enveloping shells $ZCD_{k}$ which contain singular hypersurfaces
$\mathcal{Z}_{\mathrm{sing},k}(\mathcal{I})$ and subdomains with some
number of positive and negative partition functions.
The statistical amoeba $\mathcal{A}_{k}$ is surrounded by the domain
$\mathcal{D}_{k-}$ of maximal instability (as we will demonstrate
in next section). 

For ${\displaystyle k>\left\lfloor\frac{N}{2}\right\rfloor}$ the domain $ZCD_{k}$
coincides with that of $ZCD_{N-k}$ while the domains $\mathcal{D}_{k-}$
and $\mathcal{D}_{k+}$ exchange their roles, namely $\mathcal{D}_{k-}:= \mathcal{D}_{(N-k)+}$
and $\mathcal{D}_{k+}:= \mathcal{D}_{(N-k)-}$. With increasing $k$
the stability domain schrinks (proposition \ref{prop: inclusion chain equilibrium domain}) while instable domain $\mathcal{D}_{k-}$ expands.
At $k=N$ the whole space $\mathbb{R}^{n}$ is the domain of instability.
It would be natural to refer to the domain $\mathcal{D}_{k-}\cup ZCD_{k}$
at ${\displaystyle k>\left\lfloor\frac{N}{2}\right\rfloor}$ as the statistical
$k$-antiamoeba. 

The name amoeba is borrowed from algebraic geometry. The amoeba of
the algebraic variety $\mathcal{V}_{n}$ given by the algebraic equation
\begin{equation}
\sum_{m_{1},\dots,m_{n}}c_{m_{1},\dots,m_{n}}z_{1}^{m_{1}}z_{2}^{m_{2}}\cdots z_{n}^{m_{n}}=0\label{eq: polynomial}
\end{equation}
with complex $z_{1},\dots,z_{n}$ and $c_{m_{1},\dots,m_{n}}$ is
defined \cite{GKZ2008} as the image of $\mathcal{V}_{n}$ under
logarithmic map $(z_{1},\dots,z_{n})\mapsto(\log|z_{1}|,\dots,\log|z_{n}|)$.
Amoebas of algebraic varieties and their properties have been
intensively studied since their introduction by Gelfand, Kapranov and Zelevinsky (see e.g. \cite{Purbhoo2008,Mikhalkin2000,Passare2000,Theobald2002,Mikhalkin2004,Tsikh2012,NissePassare2012}). 

In the simplest esample of the complex plane given by the equation
$1+z_{1}+z_{2}=0$ (see \cite{Mikhalkin2013}) the amoeba is defined
by the triangle inequalities 
\begin{equation}
\begin{array}{c}
e^{x}+e^{y}>1,\\
1+e^{y}>e^{x},\\
1+e^{x}>e^{y}
\end{array}\label{eq: triangle inequalities}
\end{equation}
and is presented in figure \ref{fig: examples of 1-strata, a} (colored
region). So in this case statistical and algebraic amoebas coincide.
However, in general, it is not so. Indeed let us, firstly, rewrite
equation (\ref{eq: polynomial}) in the form 
\begin{equation}
\sum_{\alpha=1}^{N}\exp\left(b_{\alpha}+\sum_{i=1}^{n}a_{\alpha i}x_{i}+\mathrm{i}\left(\arg a_{\alpha}+\sum_{i=1}^{n}a_{\alpha i}\varphi_{i}\right)\right)=0\label{eq: polynomial log coordinates}
\end{equation}
where $x_{i}=\log|z_{i}|$, $\varphi_{i}=\arg z_{i}$, $a_{\alpha i}=e^{b_{\alpha i}+\mathrm{i}\arg a_{\alpha i}}$
and rows of $a_{\alpha i}$ are given by integers $m_{i}$ ($a_{\alpha i}=m_{i}$
for given monomial indexed by $\alpha$). Projection of the $2n-2$
dimensional real hypersurface given by (\ref{eq: polynomial log coordinates})
onto the space $\mathbb{R}^{n}$ with coordinates $(x_{1},\dots,x_{n})$
is the amoeba $\mathcal{A}$ of this hypersurface \cite{Mikhalkin2000,Passare2000,Theobald2002,Mikhalkin2004,Tsikh2012,NissePassare2012,Mikhalkin2013,Mikhalkin2015}. 

On the other hand applying the triangle inequality to (\ref{eq: polynomial log coordinates}),
one gets the set of inequalities 
\begin{equation}
-e^{f_{\alpha}(\boldsymbol{x})}+\sum_{\beta\neq\alpha}e^{f_{\beta}(\boldsymbol{x})}>0\label{eq: triangle inequalities polynomial}
\end{equation}
where ${\displaystyle f_{\alpha}(\boldsymbol{x})=b_{\alpha}+\sum_{i=1}^{n}a_{\alpha i}x_{i}}$.
The domain in $\mathbb{R}^{n}$ defined by $N$ inequalities (\ref{eq: triangle inequalities polynomial})
is called approximated amoeba (non-lopsided set) $\mathcal{LA}$ \cite{Purbhoo2008}.
In general $\mathcal{LA}$ does not coincide with the amoeba $\mathcal{A}$,
namely $\mathcal{LA}\supseteq\mathcal{A}$, but ``\textit{$\mathcal{LA}$
is a very good approximation for $\mathcal{A}$}'' \cite{Purbhoo2008}. 

Comparing the set of inequalities (\ref{eq: triangle inequalities polynomial})
with our definition of the domain $\mathcal{D}_{1+}$, we can conclude
that the statistical $1$-amoeba with integer-valued elements $a_{\alpha i}$
and ${\displaystyle f_{\alpha}(\boldsymbol{x})\equiv b_{\alpha}+\sum_{i=1}^{n}a_{\alpha i}x_{i}}$
in (\ref{eq: hypersurfaces equation 1-stratum}) coincide with $\mathcal{LA}$
amoeba for the hypersurface (\ref{eq: polynomial log coordinates}).
The triangle inequalities reasoning becomes rather involved for partitions
different from $\mathcal{I}_{1}=\{\alpha\}$ and $\mathcal{I}_{2}=[N]\backslash\{\alpha\}$,
$\alpha=1,\dots,N$. 
Anyway, equilibrium domains $\mathcal{D}_{k+}$ have a simple geometrical interpretation. First, one has the following well-known  
\begin{lem}
\label{lem: closed polygonal path} If $\boldsymbol{x}$ is in the
$1$-equilibrium region $\mathcal{D}_{1+}$, then one can construct
a polygonal closed path with $N$ sides of lengths $\left(e^{f_{\alpha}(\boldsymbol{x})}:\,\alpha\in[N]\right)$ in some order. 
\end{lem}
\begin{proof}
The base case $N=3$ is equivalent to triangle
inequality. Then we assume that the assertion holds for all integers
$k$ such that $N-1\geq k\geq3$ and will proceed by
induction on $N$. One can fix $f_{1}(\boldsymbol{x})\geq f_{2}(\boldsymbol{x})\geq\dots\geq f_{N}(\boldsymbol{x})$
without loss of generality and choose a number $\ell$ such that 
\begin{equation}
{\displaystyle \max\left\{ e^{f_{2}(\boldsymbol{x})},e^{f_{1}(\boldsymbol{x})}-e^{f_{N}(\boldsymbol{x})}\right\} \leq\ell\leq\min\left\{ e^{f_{1}(\boldsymbol{x})},\sum_{\beta=2}^{N-1}e^{f_{\beta}(\boldsymbol{x})}\right\} }.\label{eq: choice auxiliary length}
\end{equation}
Note that this definition is well-posed: $e^{f_{1}(\boldsymbol{x})}-e^{f_{N}(\boldsymbol{x})}<e^{f_{1}(\boldsymbol{x})}$
since $f_{N}(\boldsymbol{x})$ is real, $e^{f_{1}(\boldsymbol{x})}-e^{f_{N}(\boldsymbol{x})}$ $\displaystyle<\sum_{\beta=2}^{N-1}e^{f_{\beta}(\boldsymbol{x})}$
since $\boldsymbol{x}$ belongs to the equilibrium region, $e^{f_{2}(\boldsymbol{x})}\leq e^{f_{1}(\boldsymbol{x})}$
by hypothesis and ${\displaystyle e^{f_{2}(\boldsymbol{x})}<\sum_{\beta=2}^{N-1}e^{f_{\beta}(\boldsymbol{x})}}$
since $N>3$. So there is at least one positive term. Let us consider
\begin{equation}
\Lambda_{1}:=\left(e^{f_{1}(\boldsymbol{x})},\ell,e^{f_{N}(\boldsymbol{x})}\right),\quad\Lambda_{2}:=\left(\ell,e^{f_{2}(\boldsymbol{x})},\dots,e^{f_{N-1}(\boldsymbol{x})}\right).\label{eq: lengths sides quasi-partition}
\end{equation}
One has ${\displaystyle e^{f_{1}(\boldsymbol{x})}=e^{f_{1}(\boldsymbol{x})}-e^{f_{N}(\boldsymbol{x})}+e^{f_{N}(\boldsymbol{x})}<\ell+e^{f_{N}(\boldsymbol{x})}}$
and $e^{f_{1}(\boldsymbol{x})}=\max\{e^{f_{1}(\boldsymbol{x})},e^{f_{N}(\boldsymbol{x})},\ell\}$
so $\Lambda_{1}$ is non-lopsided. One can see that $\ell=\max\Lambda_{2}$
and ${\displaystyle \ell<\sum_{\beta=2}^{N-1}e^{f_{\beta}(\boldsymbol{x})}}$
by construction. Then both $\Lambda_{1}$ and $\Lambda_{2}$ are non-lopsided
and their cardinalities are $3$ and $N-1$. Thus, by induction hypothesis
there exists two polygonal closed paths $\mathcal{C}_{1}$ and $\mathcal{C}_{2}$
with sides $\left(e^{f_{1}(\boldsymbol{x})},\ell,e^{f_{N}(\boldsymbol{x})}\right)$ 
and $\left(\ell,e^{f_{2}(\boldsymbol{x})},\dots,e^{f_{N-1}(\boldsymbol{x})}\right)$ in some order.
Finally, one can join $\mathcal{C}_{1}$ and $\mathcal{C}_{2}$ along
the side of length $\ell$ and get a closed polygonal path with $N$
sides of lengths $\left(e^{f_{\alpha}(\boldsymbol{x})}:\,\alpha\in[N]\right)$. 
\end{proof}

Stratification (\ref{eq: inclusion chain equilibrium domain}) can now be seen as a refinement of the triangle inequality property. 
\begin{prop}
\label{prop: constructibility and k-amoebas} The equilibrium domain
$\mathcal{D}_{k+}$ relative to the $k$-amoeba is the set of all
points $\boldsymbol{x}$ in $\mathbb{R}^{n}$ that satisfy the following condition: there exists a planar polygon with $g$ sides, for all $N-k+1\leq g\leq N$, and with
lengths of sides $\left(\ell_{1},\dots,\ell_{g}\right)$ equal to ${\displaystyle \left(\sum_{\alpha\in\mathcal{I}_{1}}e^{f_{\alpha}(\boldsymbol{x})},\dots,\sum_{\alpha\in\mathcal{I}_{g}}e^{f_{\alpha}(\boldsymbol{x})}\right)}$,
where $\left\{ \mathcal{I}_{1},\dots,\mathcal{I}_{g}\right\} $ is
any partition of $[N]$ in $g$ disjoint non-empty subsets. 
\end{prop}
\begin{proof}
Let us suppose that $\boldsymbol{x}\in\mathcal{D}_{k+}$ and consider
any $g$-partition $\left\{ \mathcal{I}_{1},\dots,\mathcal{I}_{g}\right\} $
of $[N]$. Since all subsets in $\mathcal{I}_{1},\dots,\mathcal{I}_{g}$
are not empty, then $\#\mathcal{I}_{u}\geq1$ for all $1\leq u\leq g$.
Thus 
\begin{equation}
\#\mathcal{I}_{u}=N-\sum_{w\neq u}\#\mathcal{I}_{w}\leq N-(g-1)\leq N-(N-k+1-1)=k,\quad1\leq u\leq g.\label{eq: inequalities k-constructibility}
\end{equation}
Assuming that $\boldsymbol{x}\in\mathcal{D}_{k+}$ one gets 
\[
\sum_{\alpha\in\mathcal{I}_{u}}e^{f_{\alpha}(\boldsymbol{x})}<\sum_{\beta\notin\mathcal{I}_{u}}e^{f_{\beta}(\boldsymbol{x})}=\sum_{w\neq u}\sum_{\beta\in\mathcal{I}_{w}}e^{f_{\beta}(\boldsymbol{x})},\quad1\leq u\leq g.
\]
By lemma \ref{lem: closed polygonal path}, this is equivalent to
the existence of a closed planar polygon with $g$ sides whose
lengths are (in some order) ${\displaystyle \sum_{\alpha\in\mathcal{I}_{1}}e^{f_{\alpha}(\boldsymbol{x})},\dots,\sum_{\alpha\in\mathcal{I}_{g}}e^{f_{\alpha}(\boldsymbol{x})}}$. 

Now assume that the existence hypothesis holds. In particular,
it holds at $g=N-k+1$ For any $\mathcal{I}\in\mathcal{P}_{k}[N]$
one can consider the $g$-partition $\left\{ \{\beta_{1}\},\dots,\{\beta_{g-1}\},\mathcal{I}\right\} $
where $[N]\backslash\mathcal{I}=:\{\beta_{1},\dots,\beta_{g-1}\}$.
The hypothesis in such a case implies that ${\displaystyle \sum_{\alpha\in\mathcal{I}}e^{f_{\alpha}(\boldsymbol{x})}<\sum_{u=1}^{g-1}e^{f_{\beta_{u}}(\boldsymbol{x})}=\sum_{\beta\notin\mathcal{I}}e^{f_{\beta}(\boldsymbol{x})}}$.
This means that $\mathcal{Z}_{k}(\mathcal{I};\boldsymbol{x})>0$ for
all $\mathcal{I}\in\mathcal{P}_{k}[N]$, that is $\boldsymbol{x}\in\mathcal{D}_{k+}$. 
\end{proof}

It seems that these higher amoebas, i.e. those
defined by sets of inequalities of the type (\ref{eq: triangle inequalities polynomial})
with more than one minus sign were not discussed before in this context.

Some important features of algebraic amoebas are not preserved in the general $k$-amoeba case. For example, the extremal boundary of standard $1$-amoebas is the set of points $\boldsymbol{x}\in\mathbb{R}^{n}$ such that $\boldsymbol{x}\in\mathcal{Z}_{\mathrm{sing},1}(\{\alpha\})$ for certain $\alpha\in[N]$. All non-extremal points are partitioned in two sets defined by the sign of $\mathcal{Z}_{1}(\{\alpha\})$. As already noted, such partitions at different values of $\alpha$ are compatible, in the sense that if $\mathcal{Z}_{1}(\{\alpha\})<0$ then one knows that $\mathcal{Z}_{1}(\{\beta\})>0$ for all $\beta\neq\alpha$. Conversely, crossing points in $\mathcal{Z}_{\mathrm{sing},k}(\mathcal{I}_{1})\cap\mathcal{Z}_{\mathrm{sing},k}(\mathcal{I}_{2})$ between distinct components of the $k$-singular locus at $k\geq2$ open the way for more sign combinations, see figure \ref{fig: comments higher strata, a}. 

Furthermore, each connected component of the boundary of a standard algebraic amoeba bounds a certain convex (finite or infinite) region of the space (see e.g. \cite{GKZ2008}). This property does not hold in general for $k$-singular loci $\mathcal{Z}_{\mathrm{sing},k}(\mathcal{I})$, $\mathcal{I}\in\mathcal{P}_{k}[N]$. Such a case is pointed out in figure \ref{fig: comments higher strata, b}. 

Anyway, a generalization of these properties to $k$-amoebas can be done taking in account all $k$-singular loci $\mathcal{Z}_{\mathrm{sing},k}(\mathcal{I})$ simultaneously. Thus, in our real-valued approach higher statistical amoebas arise in a natural way. 

We note also that the relations between $1$-statistical amoebas and statistical physics have been discussed in \cite{Kenyon2006} and \cite{Tsikh2009,Passare2012}.

\section{Structure of ZCD domains \label{sec: Structure of ZCDk domains}}

Zeros confinement domain $ZCD_{k}$ for $k$-stratum separating the
domains $\mathcal{D}_{k+}$ and $\mathcal{D}_{k-}$ has rather complicated
structure in general. Here we will consider some of their simplest
properties. 

For the first stratum $k=1$ the $ZCD_{1}$ collapses into the set
of hypersurfaces of zeros $\mathcal{Z}_{\mathrm{sing},1}(\{\alpha\})$.
Let $k\geq2$. Each zero hypersurface $\mathcal{Z}_{\mathrm{sing},k}(\mathcal{I}_{k})$
for given subset $\mathcal{I}_{k}$ divides the $ZCD_{k}$ in two
subdomains where the function $\mathcal{Z}_{k}(\mathcal{I}_{k};\boldsymbol{x})$
have definite, positive or negative, sign. The set of all $\binom{N}{k}$
hypersurfaces $\mathcal{Z}_{\mathrm{sing},k}(\mathcal{I}_{1})$
with all possible partitions of $[N]$ $\mathcal{I}_{1}\cup\mathcal{I}_{2}=[N]$
and cardinality $\#\mathcal{I}_{1}=k$ divides the $ZCD_{k}$ into
a finite, say $M$, number of subdomains $ZCD_{k;\delta}$ and inside
each of them each of functions $\mathcal{Z}_{k}(\mathcal{I}_{k};\boldsymbol{x})$
has definite sign. 

So, one can associate with each such subdomain $ZCD_{k;\delta}$ a
set of $\binom{N}{k}$ numbers $1$ and $-1$ coinciding with values
of sign function defined as 
\begin{equation}
s(\mathcal{I}_{1};\boldsymbol{x}):=\mbox{sign}\left[-\sum_{\alpha\in\mathcal{I}_{1}}e^{f_{\alpha}(\boldsymbol{x})}+\sum_{\beta\in[N]\backslash\mathcal{I}_{1}}e^{f_{\beta}(\boldsymbol{x})}\right]\in\{-1,0,+1\}\label{eq: sign function for subsets}
\end{equation}
evaluated for each partition $\mathcal{I}_{1}\cup\mathcal{I}_{2}=[N]$
with $\#\mathcal{I}_{1}=k$ and $\boldsymbol{x}\in\mathbb{R}^n$. If one chooses an order for the subdomains $\{ZCD_{k;\delta}\} \rightarrow [M]$ and for elements of $\mathcal{P}_{k}[N]$,
e.g. lexicographical order, then one has the set of mappings 
\begin{equation}
\left(\boldsymbol{S}_{k}(\boldsymbol{x})\right)_{\tau}:= s(\mathcal{I}_{\tau};\boldsymbol{x}),\quad\tau=1,\dots,\binom{N}{k}
\label{eq: sign vector at domain}
\end{equation}
and
\begin{equation}
\boldsymbol{S}_{k;\delta}:= \boldsymbol{S}_{k}(\boldsymbol{x})
\label{eq: sign vector at domain, 2}
\end{equation}
which assigns to a subdomain $ZCD_{k;\delta}$ a vector of $\binom{N}{k}$
components, whose $\tau$-th component is the sign of $\mathcal{Z}_{k}(\mathcal{I}_{\tau};\boldsymbol{x})$
evaluated at an interior point $\boldsymbol{x}\in ZCD_{k;\delta}$. With a slight abuse of notation, we will also use $\delta\in[M]$ to denote the corresponding $ZCD_{k;\delta}$ with the chosen order.
For example, $\boldsymbol{S}_{2;\delta}=\left(-1,1,1,1,1,-1,-1,\right.$
$\left.-1,-1,1,1,1,1,1,1\right)$ where $\{f_{\alpha}\}$ are as in
figure \ref{fig: intermediate number curves 1-stratum} and $\delta$
is the subdomain containing the point $(x,y)\equiv(2,-2)$. In the
domain $\mathcal{D}_{k+}$ one has $\boldsymbol{S}_{k;\mathcal{D}_{k+}}=(1,1,1,\dots,1)$.
Number of signs $-1$ in $\boldsymbol{S}_{k;\delta}$ varies in $ZCD_{k}\cup\mathcal{D}_{k-}$. One has 
\begin{prop}
\label{prop: bound - signs} The maximum number of $-1$ in $\boldsymbol{S}_{k;\delta}$
at fixed ${\displaystyle k<\frac{N}{2}}$ and varying $\delta\in[M]$ is equal to $\binom{N-1}{k-1}$.
If $2k=N$ then the number of $-1$ signs in $\boldsymbol{S}_{\frac{N}{2}}$
is identically equal to $\binom{2k-1}{k}$ on $\mathbb{R}^{n}\backslash\mathcal{Z}_{\mathrm{sing},\frac{N}{2}}$. 
\end{prop}
\begin{proof}
At fixed ${\displaystyle k<\frac{N}{2}}$ and for any subdomain $ZCD_{k;\delta}$
one can consider the family 
\begin{equation}
\mathcal{F}_{k;\delta-}:=\left\{ \mathcal{I}\in\mathcal{P}_{k}[N]:\,s_{k;\delta}(\mathcal{I})=-1\right\} .\label{eq: family negative}
\end{equation}
The intersection between two elements of $\mathcal{F}_{k;\delta-}$
is non-empty. Indeed, let us assume that $\mathcal{I}\in\mathcal{F}_{k;\delta-}$
and $\mathcal{I}\cap\mathcal{J}=\emptyset$, with $\#\mathcal{J}=k$.
In particular, one has $\mathcal{J}\subseteq[N]\backslash\mathcal{I}$
and $\mathcal{I}\subseteq[N]\backslash\mathcal{J}$. This implies
that 
\begin{equation}
{\displaystyle \sum_{\alpha\in[N]\backslash\mathcal{J}}e^{f_{\alpha}(\boldsymbol{x})}>\sum_{\alpha\in\mathcal{I}}e^{f_{\alpha}(\boldsymbol{x})}>\sum_{\beta\in[N]\backslash\mathcal{I}}e^{f_{\beta}(\boldsymbol{x})}\geq\sum_{\beta\in\mathcal{J}}e^{f_{\beta}(\boldsymbol{x})}}.\label{eq: chain for intersecting family -1}
\end{equation}
Thus, $\mathcal{Z}_{k}(\mathcal{J})>0$ and $\mathcal{J}\notin \mathcal{F}_{k;\delta-}$.
Then, the family $\mathcal{F}_{k;\delta-}$ of all $k$-subsets, ${\displaystyle k<\frac{N}{2}}$,
corresponding to a $-1$ sign is an intersecting family, that is a
family of subsets with same cardinality $k$ and pairwise non-empty
intersections. Hence, Erd\H{o}s\textendash Ko\textendash Rado theorem
for intersecting family (see e.g. \cite{EKR1961}) holds and so $\mathcal{F}_{k;\delta-}$
has at most $\binom{N-1}{k-1}$ elements. Moreover, this maximum is reached
exactly if all elements of the family contain a certain $\alpha_{0}\in[N]$.
This maximum is indeed attained. Let us consider
$1$-domains $D_{1-}(\alpha)$, $\alpha\in[N]$, defined as 
\begin{equation}
D_{1-}(\alpha):=\left\{ \boldsymbol{x}\in\mathbb{R}^{n}:\,\mathcal{Z}_{1}(\{\alpha\};\boldsymbol{x})=-e^{f_{\alpha}(\boldsymbol{x})}+\sum_{\beta\neq\alpha}e^{f_{\beta}(\boldsymbol{x})}<0\right\} ,\quad\alpha=1,\dots,N.\label{eq: dominant 1-cells}
\end{equation}
The linearity of functions $f_{\alpha}$ assures that not all
$D_{1-}(\alpha)$ are empty since the assumptions of proposition \ref{prop: inclusion chain equilibrium + ZCD}  are satisfied and $\vec{r}(t_{1})\in D_{1-}(\alpha_{0})$. If $\boldsymbol{x}\in D_{1-}(\alpha)$
then one has $\displaystyle |\mathcal{Z}_{k}(\mathcal{I};\boldsymbol{x})|=$\\ $\displaystyle \left|-\sum_{\alpha\in\mathcal{I}}e^{f_{\alpha}(\boldsymbol{x})}+\sum_{\beta\in[N]\backslash\mathcal{I}}e^{f_{\beta}(\boldsymbol{x})}\right|>e^{f_{\alpha}(\boldsymbol{x})}-\sum_{\beta\neq\alpha}e^{f_{\beta}(\boldsymbol{x})}>0$
for all $\mathcal{I}\in\mathcal{P}_{k}[N]$, so $D_{1-}(\alpha)\notin\mathcal{Z}_{\mathrm{sing},k}$
and all components of $\boldsymbol{S}_{k}(\boldsymbol{x})$ are not
vanishing. Let us denote a subdomain $ZCD_{k;\delta}$ such
that $D_{1-}(\alpha)\cap\delta(\alpha)\neq\emptyset$ as $\delta(\alpha)$. 
Then, $\boldsymbol{S}_{k;\delta(\alpha)}$ coincides with $\boldsymbol{S}_{k}(\boldsymbol{x})$,
$\boldsymbol{x}\in\mathcal{D}_{1-}\cap\delta(\alpha)$. The definition
of $D_{1-}(\alpha)$ implies that ${\displaystyle \mathcal{Z}_{k}(\mathcal{I};\boldsymbol{x})<0}$
if and only if $\alpha\in\mathcal{I}$ and the number of $-1$ signs
in $\boldsymbol{S}_{k;\delta(\alpha)}$ is equal to the number of $k$-subsets $\mathcal{I}\subset[N]$
containing $\alpha$, that is the maximum $\binom{N-1}{k-1}$. 

At $2k=N$ and for any point $\boldsymbol{x}\in\mathbb{R}^{n}\backslash\mathcal{Z}_{\mathrm{sing},\frac{N}{2}}$
one has $\mathcal{Z}_{\frac{N}{2}}(\mathcal{I})<0$ if and only if $\mathcal{Z}_{\frac{N}{2}}([N]\backslash\mathcal{I})>0$.
Since both $\mathcal{I}$ and $[N]\backslash\mathcal{I}$ have cardinality
${\displaystyle \frac{N}{2}}$, there is the same number of negative
and positive terms in $\boldsymbol{S}_{\frac{N}{2}}(\boldsymbol{x})$ for all $\boldsymbol{x}\in\mathbb{R}^n\backslash\mathcal{Z}_{\mathrm{sing},\frac{N}{2}}$,
that is ${\displaystyle \binom{2k-1}{k}=\binom{2k-1}{k-1}=\frac{1}{2}\binom{2k}{k}}$. 
\end{proof}
Thus, the ambient domain $\mathcal{D}_{k-}$ for each statistical
$k$-amoeba is the domain of maximal instability. Previous
proposition implies following 
\begin{cor}
\label{prop: inclusion chain instability} One has 
\begin{equation}
\mathcal{D}_{k-}\subseteq\mathcal{D}_{\hat{k}-}\label{eq: monotony instability domain}
\end{equation}
for all ${\displaystyle 1\leq k<\hat{k}<\frac{N}{2}}$. \end{cor}
\begin{proof}
Let us take $\boldsymbol{x}\in\mathcal{D}_{k-}$ at ${\displaystyle 1\leq k<\hat{k}<\frac{N}{2}}$.
From Erd\H{o}s\textendash Ko\textendash Rado theorem and proposition
\ref{prop: bound - signs} it follows that there exists $\alpha_{0}\in[N]$
such that $\mathcal{Z}_{k}(\mathcal{I};\boldsymbol{x})<0$ if and only if $\alpha_{0}\in\mathcal{I}$.
Then, let us consider $\mathcal{J}\subset[N]$ with $\#\mathcal{J}=\hat{k}$
and $\alpha_{0}\in\mathcal{J}$. One can choose a subset $\mathcal{I_{J}}\subset\mathcal{J}$
such that $\#\mathcal{I_{J}}=k$ and $\alpha_{0}\in\mathcal{I}$, thus
$\mathcal{Z}_{\hat{k}}(\mathcal{J};\boldsymbol{x})<\mathcal{Z}_{k}(\mathcal{I_{J}};\boldsymbol{x})<0$
where last inequality holds since $\alpha_{0}\in\mathcal{I_{J}}$. Then,
$\mathcal{Z}_{\hat{k}}(\mathcal{J})<0$ if $\alpha_{0}\in\mathcal{J}$ and $\mathcal{Z}_{\hat{k}}(\mathcal{J})>0$
otherwise, since additional $-1$ signs would contradict the bound $\binom{N-1}{\hat{k}-1}$
in proposition \ref{prop: bound - signs}. This means that $\boldsymbol{x}\in\mathcal{D}_{\hat{k}-}$.
\end{proof}
Consequently one also has the chain 

\begin{equation}
\mathcal{D}_{1-}\subseteq\mathcal{D}_{2-}\subseteq\cdots\subseteq\mathcal{D}_{\left(\left\lceil \frac{N}{2}\right\rceil -1\right)-}.\label{eq: inclusion chain equilibrium domain + ZCD, 2}
\end{equation}
which is dual with respect to (\ref{eq: inclusion chain equilibrium domain}).
It is equivalent to 

\begin{equation}
\mathcal{D}_{0+}\cup ZCD_{1}\supseteq\mathcal{D}_{2+}\cup ZCD_{2}\supseteq\dots\supseteq\mathcal{D}_{\left(\left\lceil \frac{N}{2}\right\rceil -1\right)+}\cup ZCD_{\left(\left\lceil \frac{N}{2}\right\rceil -1\right)}.\label{eq: inclusion chain equilibrium domain + ZCD, 1}
\end{equation}
So, the domain of complete stability $D_{k+}$ (possible) equilibrium
shrinks in transition to higher strata while the domain of instability
expands. Note that there is no domain in $\mathbb{R}^{n}$ where all
$\mathcal{Z}_{k}(\mathcal{I})$ are negative if ${\displaystyle k\leq\left\lfloor\frac{N}{2}\right\rfloor}$. 

For each subdomain $ZCD_{k;\delta}$ and the corresponding set $\boldsymbol{S}_{k;\delta}$
one can introduce also its integral characteristic 
\begin{equation}
\bar{\boldsymbol{S}}_{k;\delta}=\frac{1}{\binom{N}{k}}\cdot\sum_{\tau=1}^{\binom{N}{k}}(\boldsymbol{S}_{k;\delta})_{\tau}.\label{eq: instability energy}
\end{equation}
This quantity take values in the interval ${\displaystyle \left[1-2\frac{k}{N};1\right]}$,
${\displaystyle k\leq\left\lfloor\frac{N}{2}\right\rfloor}$. The maximum is
reached in the stability domain $\mathcal{D}_{k+}$ while the minimum
${\displaystyle -1+2\frac{k}{N}}$ is archieved at the ambient domain
$\mathcal{D}_{k-}$ of maximal instability. An absolute minimum of
$\bar{\boldsymbol{S}}_{k;\delta}$ equal to $0$ is
reached in $\mathbb{R}^n$ when $N$ is even and $N=2k$.
An example of the function $\bar{\boldsymbol{S}}_{k;\delta}$ is presented
in the figure \ref{fig: height functions}. 

\begin{figure}[tph]
\centering{
\subfigure[$\bar{\boldsymbol{S}}_{1;\delta}=\frac{1}{6}\cdot\sum_{\tau=1}^{6}(\boldsymbol{S}_{1;\delta})_{\tau}$. ]
{\includegraphics[scale=0.29]{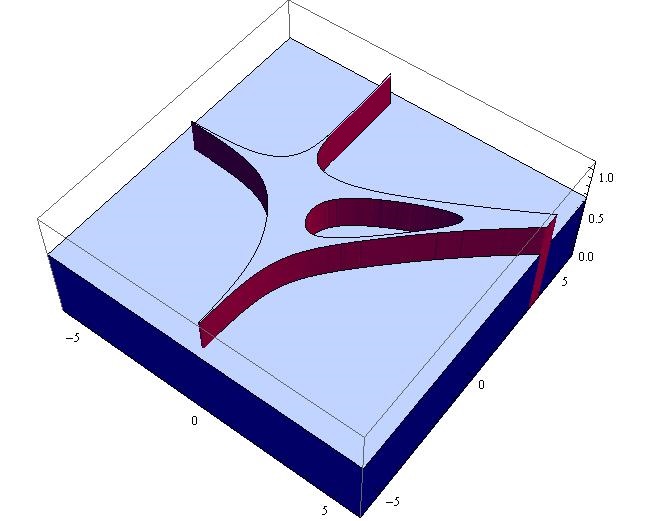}}
\hfill{}
\subfigure[$\bar{\boldsymbol{S}}_{2;\delta}=\frac{1}{15}\cdot\sum_{\tau=1}^{15}(\boldsymbol{S}_{2;\delta})_{\tau}$. ]
{\includegraphics[scale=0.27]{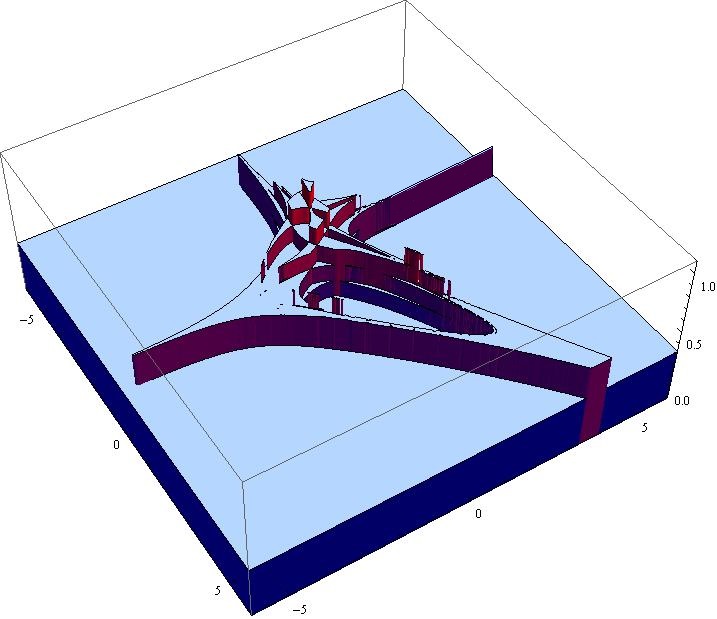}}
\caption{$\bar{\boldsymbol{S}}_{1;\delta}$ and $\bar{\boldsymbol{S}}_{2;\delta}$
in the case $f_{1}\equiv0$, $f_{2}\equiv3x$, $f_{3}\equiv3y$, $f_{4}\equiv x+y+\ln6$,
$f_{5}\equiv2x+y+\ln11$, $f_{6}\equiv x+3y+\ln4$. }
\label{fig: height functions}
}
\end{figure}

The formula (\ref{eq: instability energy}) suggests also a natural
interpretation of $\bar{\boldsymbol{S}}_{k;\delta}$. Indeed, let us view
values of sign function (\ref{eq: sign function for subsets}) as
two projections $+1$ and $-1$ of a ``spin'' associated with the subdomain $ZCD_{k;\delta}$ and certain functions $\mathcal{Z}_{k,\tau}(\boldsymbol{x})$.
So at the subdomain $ZCD_{k;\delta}$ one has a set of $\binom{N}{k}$ ``spins''
with different projections. Assuming that projections associated with
functions $\mathcal{Z}_{k,\tau}$ at different $\tau$ are realised with the
same probability ${\displaystyle w_{N,k}=\frac{1}{\binom{N}{k}}}$
then ${\displaystyle \bar{\boldsymbol{S}}_{k;\delta}}$ defined by
(\ref{eq: instability energy}) is just the mean value of spin at
the subdomain $ZCD_{k;\delta}$. 

Further, one can view the collection of $\boldsymbol{S}_{k;\delta}$
for all subdomain $ZCD_{k;\delta}$ and domains $\mathcal{D}_{k+}$,
$\mathcal{D}_{k-}$ as the set of states of the statistical system
of $\binom{N}{k}$ spins. Considering the interaction of spins with external
(magnetic) field $H$ as for the standard spin systems (see e.g. \cite{LL1980,Huang1963}),
one defines the energy 
\begin{equation}
E_{k;\delta}=-H\cdot\sum_{\tau=1}^{\binom{N}{k}}(\boldsymbol{S}_{k;\delta})_{\tau},\quad\delta=1,\dots,M.\label{eq: energy with magnetic parameter}
\end{equation}
Finally, for the partition function of the spin system one
has 
\begin{equation}
\mathcal{Z}_{k,\mathrm{spin}}=\sum_{\delta=1}^{M}\exp\left(-\beta H\cdot\sum_{\tau=1}^{\binom{N}{k}}(\boldsymbol{S}_{k;\delta})_{\tau}\right),\quad k=1,\dots,\left\lfloor\frac{N}{2}\right\rfloor \label{eq: canonical partition function spins}
\end{equation}
where $\beta$ is a parameter (say inverse of ``temperature'' $T$). 

Energy $E_{k}$ has minimum at the domain $\mathcal{D}_{k+}$ and
maximum in the domain $\mathcal{D}_{k-}$. Excited transition states
are associated with subdomains of $ZCD_{k}$. 

Introducing interaction between spins of the form ${\displaystyle E_{k;\delta,\mathrm{int}}=\gamma\cdot\sum_{\tau,\nu=1}^{\binom{N}{k}}(\boldsymbol{S}_{k;\delta})_{\tau}\cdot(\boldsymbol{S}_{k;\delta})_{\nu}}$,
one gets a partition function of the Ising type model.

\section{Tropical limit and tropical zeros \label{sec: Tropical limit and tropical zeros}}

The amoebas viewed at large distance are essentially the sets of thinning
tentacles which become certain piecewise linear objects in the tropical
limit for algebraic amoebas, see e.g. \cite{Mikhalkin2004,Maclagan2015}.
Such images of statistical $k$-amoebas are associated with the limiting
behaviour at large functions $f_{\alpha}(\boldsymbol{x})$ in the
partition function (\ref{eq: signed partition function}). In the
case of linear functions $f_{\alpha}(\boldsymbol{x})$ as in (\ref{eq: signed partition function, 2})
there are different ways to realise such a limit. The first one is
to consider large values of the variables $x_{i}$ introducing slow
variables $\tilde{x}_{i}:=\varepsilon\cdot x_{i}$, $i=1,\dots,n$
with $\varepsilon\rightarrow0$. For the $k$-th stratum the functions
$\mathcal{Z}_{k}(\mathcal{I}_{\tau};\boldsymbol{x})$ at $\varepsilon\rightarrow0$
are the superpositions of highly singular terms and the corresponding
hypersurfaces are defined as 
\begin{equation}
\sum_{\alpha=1}^{N}g_{\mathcal{I}_{\tau}(\alpha)}\exp\left(\frac{1}{\varepsilon}\cdot\sum_{i=1}^{n}a_{\alpha i}\tilde{x}_{i}\right)=0,\quad\tau=1,\dots,\binom{N}{k}\label{eq: slow variables k-stratum components}
\end{equation}
as $\varepsilon\rightarrow0$. For the first stratum ($k=1$) equations
(\ref{eq: slow variables k-stratum components}) are of the form 
\begin{equation}
\exp\left(\frac{1}{\varepsilon}\cdot\sum_{i=1}^{n}a_{\alpha i}\tilde{x}_{i}\right)=\sum_{\beta\neq\alpha}\exp\left(\frac{1}{\varepsilon}\cdot\sum_{i=1}^{n}a_{\beta i}\tilde{x}_{i}\right)\label{eq: slow variables 1-stratum}
\end{equation}
and in the limit $\varepsilon\rightarrow0$ one gets the set of hyperplanes
in $\mathbb{R}^{n}$ given by 
\begin{equation}
\sum_{i=1}^{n}a_{\alpha i}\tilde{x}_{i}=\max_{\beta\neq\alpha}\left\{ \sum_{i=1}^{n}a_{\beta i}\tilde{x}_{i}\right\} .\label{eq: slow variables 1-stratum, tropical limit}
\end{equation}
All these hyperplanes pass through the origin $\tilde{x}_{i}=0$,
$i=1,\dots,n$. They are the tropical limit of the ideal statistical
hypersurfaces considered in \cite{AK2016}. 

For higher strata and each partition $(\mathcal{I}_{1},\mathcal{I}_{2})$
the limit $\varepsilon\rightarrow0$ of equations (\ref{eq: locus signed partition}),
(\ref{eq: locus signed partition with assumption}) is given by the
set of hyperplanes 
\begin{equation}
\max_{\alpha\in\mathcal{I}_{1}}\left\{ \sum_{i=1}^{n}a_{\alpha i}\tilde{x}_{i}\right\} =\max_{\beta\in\mathcal{I}_{2}}\left\{ \sum_{i=1}^{n}a_{\beta i}\tilde{x}_{i}\right\} .\label{eq: slow variables higher strata, tropical limit}
\end{equation}

The second way to realise the limit $f_{\alpha}\rightarrow\infty$,
more close to the standard tropical limit in algebraic geometry \cite{Mikhalkin2004,Maclagan2015},
is to make the parameter $b_{\alpha}$ in $f_{\alpha}$ large too,
i.e. to consider the limit ${\displaystyle x_{i}=\frac{\tilde{x}_{i}}{\varepsilon}}$,
${\displaystyle b_{\alpha}=\frac{\tilde{b}_{\alpha}}{\varepsilon}}$,
with finite $\tilde{x}_{i}$, $\tilde{b}_{\alpha}$ and $\varepsilon\rightarrow0$.
In this case the tropical limit of the hypersurfaces (\ref{eq: locus signed partition}),
(\ref{eq: locus signed partition with assumption}) is given by 
\begin{equation}
\max_{\alpha\in\mathcal{I}_{1}}\left\{ \tilde{b}_{\alpha}+\sum_{i=1}^{n}a_{\alpha i}\tilde{x}_{i}\right\} =\max_{\beta\in\mathcal{I}_{2}}\left\{ \tilde{b}_{\beta}+\sum_{i=1}^{n}a_{\beta i}\tilde{x}_{i}\right\} .\label{eq: slow variables higher strata, tropical limit 2}
\end{equation}
Now the hyperplanes (\ref{eq: slow variables higher strata, tropical limit 2})
do not pass, in general, through the origin $\tilde{x}_{i}=0$, $i=1,\dots,n$. 

The third way is to keep variables $x_{i}$ finite, but to send to
infinity the parameters $a_{\alpha i}$ and $b_{\alpha}$ as ${\displaystyle a_{\alpha i}=\frac{\tilde{a}_{\alpha i}}{\varepsilon}}$,
${\displaystyle b_{\alpha}=\frac{\tilde{b}_{\alpha}}{\varepsilon}}$,
with $\varepsilon\rightarrow0$ and finite $\tilde{a}_{\alpha i}$,
$\tilde{b}_{\alpha}$. Such a limit of hypersurfaces (\ref{eq: locus signed partition}),
(\ref{eq: locus signed partition with assumption}) is given by the
set of hyperplanes defined by equations 
\begin{equation}
\max_{\alpha\in\mathcal{I}_{1}}\left\{ \tilde{b}_{\alpha}+\sum_{i=1}^{n}\tilde{a}_{\alpha i}x_{i}\right\} =\max_{\beta\in\mathcal{I}_{2}}\left\{ \tilde{b}_{\beta}+\sum_{i=1}^{n}\tilde{a}_{\beta i}x_{i}\right\} .\label{eq: slow variables higher strata, tropical limit 3}
\end{equation}
Equations (\ref{eq: slow variables higher strata, tropical limit 2})
and (\ref{eq: slow variables higher strata, tropical limit 3}) are
related via exchange $a_{\alpha i}\leftrightarrow\tilde{a}_{\alpha i}$,
$\tilde{x}_{i}\leftrightarrow x_{i}$ keeping in both cases the product
${\displaystyle a_{\alpha i}x_{i}\sim\frac{1}{\varepsilon}}$. 

For different strata the sets of equations (\ref{eq: slow variables higher strata, tropical limit 2})
or (\ref{eq: slow variables higher strata, tropical limit 3}), defining
the tropical limit of hypersurfaces (\ref{eq: locus signed partition}),
(\ref{eq: locus signed partition with assumption}) are quite different.
However, one has 
\begin{prop}
\label{prop: Tropical limit strata} In the tropical limits considered
above, zeros loci of $\mathcal{Z}_{k,\mathrm{trop}}(\mathcal{I})$
given by equations (\ref{eq: slow variables higher strata, tropical limit 2})
or (\ref{eq: slow variables higher strata, tropical limit 3}) are
the same for all strata. All domains $ZCD_{k}$ collapse into a single
set of piecewise hyperplanes given e.g. by equations (\ref{eq: slow variables higher strata, tropical limit 2})
or (\ref{eq: slow variables higher strata, tropical limit 3}) for
the first stratum $\mathcal{Z}_{\mathrm{sing},1}$. \end{prop}
\begin{proof}
Let us denote $\mathcal{I}_{1}:=\mathcal{I}$ and $\mathcal{I}_{2}:=[N]\backslash\mathcal{I}$. The equation ${\displaystyle \mathcal{Z}_{k}(\mathcal{I}_{1};\boldsymbol{x})=0}$,
$\#\mathcal{I}_{1}=k$, is equivalent to ${\displaystyle \sum_{\alpha\in\mathcal{I}_{1}}e^{f_{\alpha}(\boldsymbol{x})}=\sum_{\beta\in\mathcal{I}_{2}}e^{f_{\beta}(\boldsymbol{x})}}$.
In term of slow variables it becomes $\displaystyle \sum_{\alpha\in\mathcal{I}_{1}}\exp\left(\frac{f_{\alpha}(\boldsymbol{x})}{\varepsilon}\right)=$ $\displaystyle\sum_{\beta\in\mathcal{I}_{2}}\exp\left(\frac{f_{\beta}(\boldsymbol{x})}{\varepsilon}\right)$. Let us take $\bar{\alpha}_{i}\in\mathcal{I}_{i}$
such that $f_{\bar{\alpha}_{i}}(\boldsymbol{x})={\displaystyle \max_{\alpha\in\mathcal{I}_{i}}\{f_{\alpha}(\boldsymbol{x})\}}$,
$i\in\{1,2\}$. Then, previous equation is equivalent to 
\begin{equation}
\exp\frac{f_{\bar{\alpha}_{1}}(\boldsymbol{x})}{\varepsilon}\cdot\left[\sum_{\beta\in\mathcal{I}_{1}}\exp\frac{f_{\beta}(\boldsymbol{x})-f_{\bar{\alpha}_{1}}(\boldsymbol{x})}{\varepsilon}\right]=\exp\frac{f_{\mathcal{\bar{\alpha}}_{2}}(\boldsymbol{x})}{\varepsilon}\cdot\left[\sum_{\gamma\in\mathcal{I}_{2}}\exp\frac{f_{\gamma}(\boldsymbol{x})-f_{\bar{\alpha}_{2}}(\boldsymbol{x})}{\varepsilon}\right].
\label{eq: slow variables with explicit dominants}
\end{equation}
Both the factors in square bracket in (\ref{eq: slow variables with explicit dominants}) lie in the interval $[1,N-k]$ independently on $\varepsilon\in\mathbb{R}_{+}$. Hence they are finite and non-vanishing. Thus, $\displaystyle \exp\frac{f_{\mathcal{I}_{1}}(\boldsymbol{x})-f_{\mathcal{I}_{2}}(\boldsymbol{x})}{\varepsilon}$ lies in $\displaystyle \left[\frac{1}{N-k};N-k\right]$ for all $\varepsilon\in\mathbb{R}_{+}$. Considering the limit $\varepsilon\rightarrow 0$ one gets
\begin{equation}
\displaystyle \max_{\alpha\in\mathcal{I}_{1}}\{f_{\alpha}(\boldsymbol{x})\}=\max_{\beta\in\mathcal{I}_{2}}\{f_{\beta}(\boldsymbol{x})\}.\label{eq: coincident maxima partitions}
\end{equation}
Note that $\bar{\alpha}_{1}\neq\bar{\alpha}_{2}$
since they belong to different parts of the partition. For any
$\gamma\in[N]$, $\gamma\in\mathcal{I}_{i}$ for exactly one $i\in\{1,2\}$,
so ${\displaystyle f_{\gamma}(\boldsymbol{x})\leq\max_{\alpha\in\mathcal{I}_{1}}\{f_{\alpha}(\boldsymbol{x})\}=f_{\bar{\alpha}_{i}}(\boldsymbol{x})}$.
Thus, ${\displaystyle f_{\bar{\alpha}_{1}}(\boldsymbol{x})=f_{\bar{\alpha}_{2}}(\boldsymbol{x})=\max_{\gamma\in[N]}\{f_{\gamma}(\boldsymbol{x})\}}$
so the maximum ${\displaystyle \max_{\gamma\in[N]}\{f_{\gamma}(\boldsymbol{x})\}}$
is attained at least twice, once for each index $i\in\{1,2\}$ of
$\mathcal{I}_{i}$. Considering all such partitions with $\#\mathcal{I}=k$, one gets the union
of all these tropical loci. This is the set of all points $\boldsymbol{x}\in\mathbb{R}^{n}$
such that maximum of $\{f_{1}(\boldsymbol{x}),\dots,f_{N}(\boldsymbol{x})\}$
is attained at least twice and it is independent of the stratum $k$
considered. 
\end{proof}
So, in the tropical limit the statistical amoebas collapse into the
$(n-1)$-dimensional objects $\mathcal{A}_{\mathrm{trop}}$ formed
by pieces of hyperplanes and the maximal instability domais $\mathcal{D}_{k-}$
expand to the almost whole space $\mathbb{R}^{n}$, namely to $\mathbb{R}^{n}\backslash\mathcal{A}_{\mathrm{trop}}$.
Points of the piecewise hyperplanes $\mathcal{A}_{\mathrm{trop}}$
are tropical zeros of partition function. 

These different kinds of tropical limit provide different structures for the same underlying model. For example, the tropical limit of the first kind (\ref{eq: slow variables higher strata, tropical limit}) highlights the degree $1$ homogeneous part of
linear functions $f_{\alpha}$. More in general, it gives the dominant
homogeneous parts of functions $f_{\alpha}$ and can be applied in
the study of emergence of degenerate metrics from tropical limit,
see e.g. \cite{AK2016}. An advantage of tropical limit of the first
kind is that it has a rather simple geometry. 
\begin{lem}
\label{lem: power sums for tropical I} If
$f_{\alpha}$ are $N$ real functions then ${\displaystyle \sum_{\alpha=1}^{N}e^{\lambda\cdot f_{\alpha}(\boldsymbol{x})}\leq\left(\sum_{\alpha=1}^{N}e^{f_{\alpha}(\boldsymbol{x})}\right)^{\lambda}}$
for all $\lambda\geq1$.
\end{lem}
\begin{proof}
For all $\lambda\geq1$ one has 
\begin{equation}
{\displaystyle 0<\left(\frac{e^{f_{\alpha}(\boldsymbol{x})}}{\sum_{\beta}e^{f_{\beta}(\boldsymbol{x})}}\right)^{\lambda}\leq\frac{e^{f_{\alpha}(\boldsymbol{x})}}{\sum_{\beta}e^{f_{\beta}(\boldsymbol{x})}}<1},\quad\alpha\in[N],\label{eq: tropical limit first kind, inequality 1}
\end{equation}
which implies 
\begin{equation}
0<\sum_{\alpha=1}^{N}\left(\frac{e^{f_{\alpha}(\boldsymbol{x})}}{\sum_{\beta}e^{f_{\beta}(\boldsymbol{x})}}\right)^{\lambda}\leq\sum_{\alpha=1}^{N}\frac{e^{f_{\alpha}(\boldsymbol{x})}}{\sum_{\beta}e^{f_{\beta}(\boldsymbol{x})}}=1\Rightarrow\sum_{\alpha=1}^{N}e^{\lambda f_{\alpha}(\boldsymbol{x})}\leq\left(\sum_{\alpha=1}^{N}e^{f_{\alpha}(\boldsymbol{x})}\right)^{\lambda}.\label{eq: tropical limit first kind, inequality 2}
\end{equation}
\end{proof}
\begin{prop}
\label{prop: unbounded connected components tropical graph}Connected
components of the complement of the tropical graph of the first kind
are unbounded. For homogeneous functions, connected components of $\mathcal{D}_{1-}$ are unbounded too. \end{prop} 
\begin{proof} Given $N$ linear functions $f_{1},\dots,f_{N}$, let $\varphi_{\alpha}(\boldsymbol{x})=f_{\alpha}(\boldsymbol{x})-f_{\alpha}(\boldsymbol{0})$ be the $1$-homogeneous part of $f_{\alpha}$, ${\displaystyle \Delta_{1-}(\alpha):=\left\{ \boldsymbol{y}\in\mathbb{R}^{n}:\,e^{\varphi_{\alpha}(\boldsymbol{y})}>\sum_{\beta\neq\alpha}e^{\varphi_{\beta}(\boldsymbol{y})}\right\} }$ be the instability domain where $\varphi_{\alpha}$ dominates and ${\displaystyle \Delta_{1-}^{\mathrm{trop}}(\alpha):=\left\{ \boldsymbol{y}\in\mathbb{R}^{n}:\,\varphi_{\alpha}(\boldsymbol{y})>\max_{[N]\backslash\{\alpha\}}\{\varphi_{\beta}(\boldsymbol{y})\}\right\} }$ be the tropical limit of $\Delta_{1-}(\alpha)$. 
In particular, it easily follows from the definitions that $\Delta_{1-}(\alpha)\subseteq\Delta_{1-}^{\mathrm{trop}}(\alpha)$. 
From (\ref{eq: slow variables higher strata, tropical limit}), $\left\{ f_{\alpha}\right\} $ and $\left\{ \varphi_{\alpha}\right\} $ have the same tropical limit of the first kind, so we focus on the latter set of functions. If $\boldsymbol{x}\in\Delta_{1-}(\alpha)$ and $\lambda\geq1$ then 
\begin{equation} 
e^{\varphi_{\alpha}(\lambda\boldsymbol{x})}=\left(e^{\varphi_{\alpha}(\boldsymbol{x})}\right)^{\lambda}>\left(\sum_{\beta\neq\alpha}e^{\varphi_{\beta}(\boldsymbol{x})}\right)^{\lambda}\geq\sum_{\beta\neq\alpha}e^{\lambda\varphi_{\beta}(\boldsymbol{x})}=\sum_{\beta\neq\alpha}e^{\varphi_{\beta}(\lambda\boldsymbol{x})}
\label{eq: tropical limit first kind, inequality prop} 
\end{equation} 
where the second inequality follow from lemma \ref{lem: power sums for tropical I} applied to homogeneous functions $\varphi_{\beta}$, $\beta\neq\alpha$. Hence $\lambda\cdot\boldsymbol{x}\in\Delta_{1-}(\alpha)$ for all $\boldsymbol{x}\in\Delta_{1-}(\alpha)$ and $\lambda\geq1$. In the same way one can show that $\lambda\cdot\boldsymbol{x}\in\Delta_{1-,\mathrm{trop}}(\alpha)$ for all $\boldsymbol{x}\in\Delta_{1-,\mathrm{trop}}(\alpha)$ and $\lambda\geq1$. 
So, let $\mathcal{C}$ (respectively, $\mathcal{C}^{\star}$) be a connected component of $\Delta_{1-}(\alpha)$ (respectively, of $\Delta_{1-}^{\mathrm{trop}}(\alpha)$) and choose $\boldsymbol{x}\in \mathcal{C}$ (respectively, $\boldsymbol{x^{\star}}\in\mathcal{C}^{\star}$). One has $\{\lambda\cdot \boldsymbol{x}:\,\lambda\geq 1\}\subseteq\mathcal{C}$ since the ray $\{\lambda\cdot \boldsymbol{x}:\,\lambda\geq 1\}$ is a connected subset of $\Delta_{1-}(\alpha)$ intersecting $\mathcal{C}$ and $\mathcal{C}$ is maximal among connected subsets of $\Delta_{1-}(\alpha)$. Similarly, $\{\lambda\cdot \boldsymbol{x^{\star}}:\,\lambda\geq 1\}\subseteq\mathcal{C}^{\star}$. Thus both $\mathcal{C}$ and $\mathcal{C}^{\star}$ are unbounded since they contain an unbounded subset. 
\end{proof}
Thus, tropical limit of the first kind has simple topological properties. For example, in two-dimensional case, the result of proposition \ref{prop: unbounded connected components tropical graph} means a trivial homotopy
for the resulting tropical graph. 

It is worth mentioning that terms in (\ref{eq: slow variables higher strata, tropical limit}) coincide with $f_{\alpha}(\boldsymbol{x})$
if $f_{\alpha}(\boldsymbol{0})=0$ for all $\alpha\in[N]$. Homogeneous linear functions ${\displaystyle f_{\alpha}(\boldsymbol{x})\equiv\sum_{i=1}^{n}\kappa_{\alpha}^{i}x_{i}}$
with real distinct parameters $\kappa_{1}<\dots<\kappa_{N}$ represent a particular example. These
functions arise in the study of Wronskian soliton solutions of KP
II equation where $e^{f_{\alpha}(\boldsymbol{x})}$ are special solutions
of the heat hierarchy. If one considers the tropical limit of the second kind (\ref{eq: slow variables higher strata, tropical limit 2})
instead of (\ref{eq: slow variables higher strata, tropical limit}), then the resulting object has a more refined
structure and many combinatorial properties (see e.g. \cite{KodamaWilliams2013}). 

Tropical limits discussed above are quite meaningful in the
statistical physics of macrosystems. Tropical limit of free energy
considered in \cite{AK2015} corresponds to $n=1$, ${\displaystyle \tilde{x}_{1}=\frac{1}{k_{B}T}}$,
$a_{\alpha1}=-E_{\alpha}$, $\displaystyle b_{\alpha}=\frac{S_{\alpha}}{k_{B}}$,
$\varepsilon=k_{B}$ where $T$ is the temperature, $\{E_{\alpha}\}$
is the energy spectrum, ${\displaystyle \exp\left(\frac{S_{\alpha}}{k_{B}}\right)}$
are degenerations of energy levels and $k_{B}$ is the Boltzmann constant. 

One can consider also more complicated situations when some of the
products $a_{\alpha i}\cdot x_{i}$ remain finite, for instance, when ${\displaystyle x_{i_{0}}=\frac{\tilde{x}_{i_{0}}}{\varepsilon}}$
and $a_{\alpha i_{0}}=\varepsilon\cdot\tilde{a}_{\alpha i_{0}}$.
In such a case the product $a_{\alpha i_{0}}\cdot x_{i_{0}}$ does
not contribute in the limit $\varepsilon\rightarrow0$ and the corresponding
equation (\ref{eq: slow variables higher strata, tropical limit}),
or (\ref{eq: slow variables higher strata, tropical limit 2}), will
not contain the variable $\tilde{x}_{i_{0}}$. So in the tropical
limit the zero locus is a piecewise hyperplane of cylindrical type. 

Such non-uniform scaling behaviour of the variables $x_{i}$ or parameters
$a_{\alpha i}$ and its connections with the multiscale tropical limit
will be discussed elsewhere.

\section{Conclusion\label{sec: Conclusion}}

In this paper partition functions (\ref{eq: signed partition function})
with linear $f_{\alpha}(\boldsymbol{x})$ have been studied. The case
of nonlinear functions $f_{\alpha}(\boldsymbol{x})$ is of great interest
too. Many general properties of singular sectors described above,
e.g. stratification of statistical $k$-amoebas, remain unchanged for more general polynomial
functions $f_{\alpha}(\boldsymbol{x})$. Specifically, for polynomials $f_{\alpha}(\boldsymbol{x})$ the set of roots (\ref{eq: set of roots lemma}) is finite and at least one $\mathcal{D}_{1-}(\alpha)$ in (\ref{eq: dominant 1-cells}) is not empty. These hypotheses are crucial for propositions as \ref{prop: inclusion chain equilibrium + ZCD}, \ref{prop: bound - signs} and corollary \ref{prop: inclusion chain instability} to be valid. Proposition \ref{prop: unbounded connected components tropical graph} can be generalized to polynomials by considering their degree $d$ homogeneous parts, where $\displaystyle d:=\max_{\alpha\in[N]}\{\deg f_{\alpha}\}<\infty$. Figure \ref{fig: polynomial stratification} presents an example of such a type. 

\begin{figure}[tph]
\centering{
\includegraphics[scale=0.4]{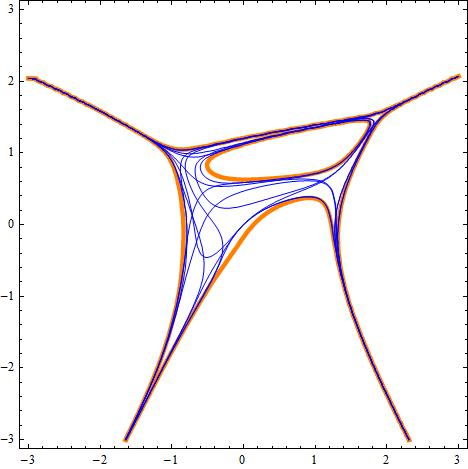}
\caption{Stratification of $\mathcal{Z}_{\mathrm{sing},2}$ (blue) and $\mathcal{Z}_{\mathrm{sing},1}$
(orange) in a nonlinear polynomial case: $f_{1}\equiv 0$, $f_{2}\equiv 3x^{2}$,
$f_{3}\equiv 3y^{3}$, $f_{4}\equiv x+xy+\ln 6$, $f_{5}\equiv 2x+y^{2}+\ln 11$,
$f_{6}\equiv x+3y+xy+\ln 4$. }
\label{fig: polynomial stratification}
}
\end{figure} 

However for general nonlinear functions situation is quite different. An example with non-polynomial functions 
\begin{equation}
f_{\alpha}(\boldsymbol{x})\equiv\left\{ \begin{array}{c}
c_{\alpha}\cdot\eta(||\boldsymbol{x}||-\frac{\alpha+1999}{1000})\quad\frac{\alpha+1999}{1000}>||\boldsymbol{x}||\\
d_{\alpha}\cdot\eta(||\boldsymbol{x}||-\frac{\alpha+1999}{1000}),\quad||\boldsymbol{x}||\geq\frac{\alpha+1999}{1000}
\end{array}\right.\label{eq: out of stratification functions example}
\end{equation}
where ${\displaystyle \eta(z)=\left\{ \begin{array}{c}
1-\exp\left(-\frac{z^{2}}{1-z^{2}}\right),\quad 1>|z|\\
1,\quad|z|\geq 1
\end{array}\right.}$, $(c_{1},d_{1})=(\ln20,\ln8)$, $(c_{2},d_{2})=(\ln20,\ln2)$ and
$(c_{\alpha},d_{\alpha})=(\ln2,\ln2)$, $\alpha=3,\dots,10$, is shown
in figure \ref{out of stratification}. Let us consider $S_{k}(\boldsymbol{x}):= C^{10}_k\cdot \bar{\boldsymbol{S}}_{k}(\boldsymbol{x})$ at $k=3,4$. At $||\boldsymbol{x}||\leq 1$ one has $\mathcal{Z}_{k}(\mathcal{I};\boldsymbol{x})<0$ iff $\{1,2\}\subset\mathcal{I}$. Thus $S_{k}(\boldsymbol{x})=\binom{10}{k}-2\cdot \binom{8}{k-2}$. At $1\leq||\boldsymbol{x}||< 2$ $S_{k}(\boldsymbol{x})$ is not decreasing. At $||\boldsymbol{x}||\geq 2$ one has $\mathcal{Z}_{3}(\mathcal{I};\boldsymbol{x})>0$ for all $\mathcal{I}\in\mathcal{P}_{3}[10]$, then $S_3(\boldsymbol{x})=C^{10}_3$. On the other hand, at $||\boldsymbol{x}||\geq 4$ one has $\mathcal{Z}_{4}(\mathcal{I};\boldsymbol{x})<0$ iff $1\in\mathcal{I}$, hence $S_{4}(\mathcal{I};\boldsymbol{x})$ has a minimum according to Erdos-Ko-Rado theorem \cite{EKR1961}. In conclusion, $\displaystyle \min_{\boldsymbol{x}}S_{3}(\boldsymbol{x})=\binom{10}{3}-2\cdot \binom{8}{1}=104$ is attained only if $||\boldsymbol{x}||< 2$. Vice versa, $\displaystyle \min_{\boldsymbol{x}}S_{4}(\boldsymbol{x})=\binom{10}{4}-2\cdot \binom{9}{3}=42$ is attained only if $||\boldsymbol{x}||>2$. In particular, $\mathcal{D}_{3-}\nsubseteq\mathcal{D}_{4-}$. 

\begin{figure}[tph]
\centering{
\subfigure[Functions $f_{\alpha}(\boldsymbol{x})$, $\alpha=1,\dots,10$, in
(\ref{eq: out of stratification functions example}).]
{\includegraphics[scale=0.28]{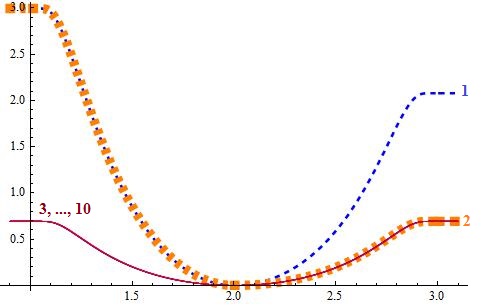}}
\hfill{}
\subfigure[Detail of functions $f_{\alpha}(\boldsymbol{x})$ in (\ref{eq: out of stratification functions example}),
$\alpha=3,\dots,10$.]
{\includegraphics[scale=0.28]{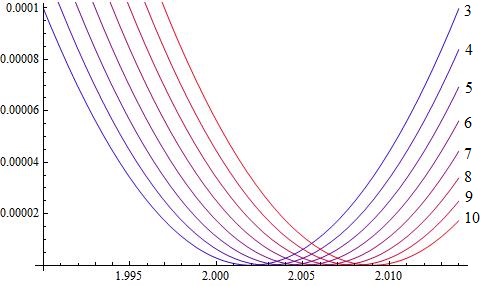}}
\hfill{}
\subfigure[Instability domain $\mathcal{D}_{3-}$ and $\mathcal{D}_{4-}$ are
highlighted.]
{\includegraphics[scale=0.28]{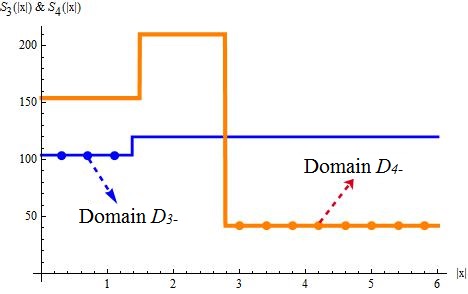}}
\caption{A non-polynomial case when chain stratification of instability domains fails.}
\label{out of stratification}
}
\end{figure}

Singular sectors of partition functions (\ref{eq: signed partition function})
with nonlinear $f_{\alpha}(\boldsymbol{x})$ will be considered in
a separate publication. \vskip6pt

\enlargethispage{20pt}


\begin{thebibliography}{9}


\bibitem{LL1980} 
	L. D. Landau and E. M. Lifschitz,  
	\textit{Statistical Physics},
	Course of Theoretical Physics, Part 1, Vol. 5 
	(Butterworth-Heinemann, 1980).

\bibitem{Huang1963}
	K. Huang,  
	\textit{Statistical mechanics}
	(New York: John Wiley \& Sons, 1963).

\bibitem{LeeYang1952a}
	C. N. Yang and T. D. Lee,  
	"Statistical theory of equations of state and phase transitions. I. Theory of condensation", 
	\href{http://dx.doi.org/10.1103/physrev.87.404}{Phys. Rev.} \textbf{87}(3), 
	404-409
	(1952). 

\bibitem{LeeYang1952b}
	T. D. Lee and C. N. Yang, 
	"Statistical theory of equations of state and phase transitions. II. Lattice gas and Ising model",
	\href{http://dx.doi.org/10.1103/physrev.87.410}{Phys. Rev.}, 
	\textbf{87}(3),
	410-419 
	(1952).

\bibitem{Fisher1965}
	M. E. Fisher,
	In \textit{Lectures in Theoretical Physics}, Vol. 7C, Chap. 1
	(W. E. Brittin (Ed.), Boulder: University of Colorado Press, 1965)


\bibitem{KG1971}
	P. J. Kortman and R. B. Griffiths, 
	"Density of Zeros on the Lee-Yang Circle for Two Ising Ferromagnets",  
	\href{http://dx.doi.org/10.1103/physrevlett.27.1439}{Phys. Rev. Lett.} \textbf{27}(21), 
	1439-1442
	(1971).

\bibitem{FS1971}
	M. Suzuki and M. E. Fisher,
	"Zeros of the partition function for the Heisenberg, ferroelectric, and general Ising models", 
	\href{http://dx.doi.org/10.1063/1.1665583}{J. Math. Phys.} \textbf{12}(2), 
	235-246, 
	(1971)

\bibitem{Ruelle1973}
	D. Ruelle,
	"Some remarks on the location of zeroes of the partition function for lattice systems", 
	\href{http://dx.doi.org/10.1007/bf01646488}{Comm. Math. Phys.} \textbf{31}, 
	265-277
	(1973).

\bibitem{Derrida1981}
	B. Derrida,
	"Random-energy model: An exactly solvable model of disordered systems", 
	\href{http://dx.doi.org/10.1103/physrevb.24.2613}{Phys. Rev. B.} \textbf{24}(5), 
	 2613-2626 
	(1981).

\bibitem{Lieb1981}
	E. H. Lieb and A. D. Sokal,  
	"A general Lee-Yang theorem for one-component and multicomponent ferromagnets",  
	\href{http://dx.doi.org/10.1007/bf01213009}{Commun. Math. Phys.} \textbf{80}(2), 
	153-179
	(1981).

\bibitem{Pearson1982}
	R. B. Pearson, 
	"Partition function of the Ising model on the periodic $4\times 4\times 4$ lattice", 
	\href{http://dx.doi.org/10.1103/physrevb.26.6285}{Phys. Rev. B} \textbf{26}(11), 
	6285-6290 
    (1982).

\bibitem{Derrida1983}
	B. Derrida, L. De Seze and C. Itzykson,   
	"Fractal structure of zeros in hierarchical models",  
	\href{http://dx.doi.org/10.1007/bf01018834}{J. Stat. Phys.} \textbf{33}(3), 
	559-569 
	(1983).

\bibitem{Borgs1990}
	C. Borgs and R. Kotecký,  
	"A rigorous theory of finite-size scaling at first-order phase transitions"  
	\href{http://dx.doi.org/10.1007/bf01013955}{J. Stat. Phys.} \textbf{61}(1-2), 
	79-119 
	(1990).

\bibitem{Biskup2004}
	M. Biskup, C. Borgs, J. T. Chayes, L. J. Kleinwaks, R. Kotecký, 
	"Partition function zeros at first-order phase transitions: A general analysis" 
	\href{http://dx.doi.org/10.1007/s00220-004-1169-5}{Commun. Math. Phys.} \textbf{251}, 
	79-131 (2004).

\bibitem{Wei2012}
	B.-B. Wei and R.-B. Liu,
	"Lee-Yang zeros and critical times in decoherence of a probe spin coupled to a bath" 
	\href{http://dx.doi.org/10.1103/physrevlett.109.185701}{Phys. Rev. Lett.} \textbf{109}(18),
	(2012).

\bibitem{Takahashi2012}
	T. Obuchi and K. Takahashi, 
	"Dynamical singularities of glassy systems in a quantum quench",  
	\href{http://dx.doi.org/10.1103/physreve.86.051125}{Phys. Rev. E} \textbf{86}(5), 
	(2012).

\bibitem{Wei2014}
	B.-B. Wei, S.-W. Chen, H.-C. Po, R.-B. Liu,  
	"Phase transitions in the complex plane of physical parameters", 
	\href{http://dx.doi.org/10.1038/srep05202}{Sci. Rep.} \textbf{4}, 
	(Nature Publishing Group, 2014).

\bibitem{Peng2015}
	X. Peng, H. Zhou, B.-B. Wei, J. Cui, J. Du, R.-B. Liu,  
	"Experimental observation of Lee-Yang zeros", 
	\href{http://dx.doi.org/10.1103/physrevlett.114.010601}{Phys. Rev. Lett.} \textbf{114}(1), 
	(2015). 

\bibitem{Langer1969}
	J. S. Langer, 
	"Statistical theory of the decay of metastable states", 
	\href{http://dx.doi.org/10.1016/0003-4916(69)90153-5}{Ann. Phys.} \textbf{54}(2), 
	258-275 (1969).

\bibitem{Newman1980}
	C. M. Newman and L. S. Schulman,  
	"Complex free energies and metastable lifetimes", 
	\href{http://dx.doi.org/10.1007/bf01012588}{J. Stat. Phys.} \textbf{23}(2),
	131-148 
	(1980).

\bibitem{Parisi1980}
	G. Parisi,  
	"A sequence of approximated solutions to the SK model for spin glasses",  
	\href{http://dx.doi.org/10.1088/0305-4470/13/4/009}{J. Phys. A: Math. Gen.} \textbf{13}(4), 
	L115-121
	(1980).

\bibitem{MPV1987}
	M. Mezard, G. Parisi and M. A. Virasoro,  
	\textit{Spin glass theory and beyond},
	(Singapore: World Scientific, 1987). 

\bibitem{Nishimori1988}
	Y. Ozeki Y and H. Nishimori,  
	"Distribution of Yang-Lee Zeros of the $\pm J$ Ising Model",  
	\href{http://dx.doi.org/10.1143/jpsj.57.1087}{J. Phys. Soc. Jpn.} \textbf{57}(3), 
	1087-1093 (1988).

\bibitem{Derrida1991}
	B. Derrida,  
	"The zeroes of the partition function of the random energy model" 
	\href{http://dx.doi.org/10.1016/0378-4371(91)90130-5}{Physica A} \textbf{177}(1-3), 
	31-37 (1991).

\bibitem{Bhanot1993}
	G. Bhanot and J. Lacki,  
	"Partition function zeros and the three-dimensional Ising spin glass", 
	\href{http://dx.doi.org/10.1007/bf01048099}{J. Stat. Phys.} \textbf{71}(1-2), 
	259-267 (1993).

\bibitem{Matsuda2008}
	Y. Matsuda, H. Nishimori and K. Hukushima,  
	"The distribution of Lee-Yang zeros and Griffiths singularities in the $\pm J$ model of spin glasses",  
	\href{http://dx.doi.org/10.1088/1751-8113/41/32/324012}{J. Phys. A: Math. Theor.} \textbf{41}(32), 
	324012 (2008).

\bibitem{Obuchi2012}
	T. Obuchi and K. Takahashi,  
	"Partition-function zeros of spherical spin glasses and their relevance to chaos",   
	\href{http://dx.doi.org/10.1088/1751-8113/45/12/125003}{J. Phys. A: Math. Theor.} \textbf{45}(12), 
	125003, (2012).

\bibitem{Feynman1987}
	R. P. Feynman,  
	"Negative probability", in \textit{Quantum implications: essays in honour of David Bohm}, 
	235-248 
	(1987).

\bibitem{Blizard1990}
	W. D. Blizard,  
	"Negative membership", 
	\href{http://dx.doi.org/10.1305/ndjfl/1093635499}{Notre Dame J. Formal Logic} \textbf{31}(3), 
	346-368 
	(1990).

\bibitem{Burgin2010}
	M. Burgin, 
	"Interpretations of negative probabilities",  
	arXiv preprint, 
	\url{arXiv:1008.1287}.

\bibitem{Purbhoo2008}
	K. Purbhoo,
	"A Nullstellensatz for amoebas", 
	\href{http://dx.doi.org/10.1215/00127094-2007-001}{Duke. Math. J.} \textbf{141}(3), 
	407-445
	(2008)

\bibitem{Tsikh2009}
	D. Y. Pochekutov and A. K. Tsikh, 
	"On the Asymptotics of Laurent Coefficients and its Application in Statistical Mechanics", 
	J. Siberian Fed. Univ.: Math. Phys. \textbf{2}(4), 
	483-493 
	(2009). 

\bibitem{Kapranov2011}
	M. Kapranov,  
	"Thermodynamics and the moment map", 
	arXiv preprint, 
	\url{arXiv:1108.3472} 
	(2011).

\bibitem{Passare2012}
	M. Passare, D. Pochekutov, A. Tsikh,  
	"Amoebas of complex hypersurfaces in statistical thermodynamics"  
	\href{http://dx.doi.org/10.1007/s11040-012-9122-x}{Math. Phys. Anal. Geom.} \textbf{16}(1), 
	89-108 (2012).

\bibitem{Adler1980}
	M. Adler and P. Van Moerbeke, 
	"Completely integrable systems, Euclidean Lie algebras, and curves",   
	\href{http://dx.doi.org/10.1016/0001-8708(80)90007-9}{Adv. Math.} \textbf{38}(3), 
	267-317 
	(1980).

\bibitem{Kodama2006}
	L. Casian and Y. Kodama,  
	"Toda lattice, cohomology of compact Lie groups and finite Chevalley groups", 
	\href{http://dx.doi.org/10.1007/s00222-005-0492-6}{Invent. Math.} \textbf{165}(1), 
	163-208 (2006).

\bibitem{AK2016}  
	M. Angelelli and B. Konopelchenko, 
	"Geometry of the basic statistical physics mapping",  
	\href{http://dx.doi.org/10.1088/1751-8113/49/38/385202}{J. Phys. A: Math. Theor.} \textbf{49}(38), 
	385202 
	(2016). 

\bibitem{GKZ2008}
	I. M. Gelfand, M. M. Kapranov and A. V. Zelevinsky,  
	\textit{Discriminants, Resultants, and Multidimensional Determinants},  
	Modern Birkh\"{a}user Classics 
	(Boston, MA: Birkh\"{a}user, 2008).

\bibitem{Mikhalkin2000}
	G. Mikhalkin, 
	"Real algebraic curves, the moment map and amoebas",  
	\href{http://dx.doi.org/10.2307/121119}{Ann. Math.} \textbf{151}(1), 
	309-326 (2000).

\bibitem{Passare2000}
	M. Forsberg, M. Passare and A. Tsikh,  
	"Laurent determinants and arrangements of hyperplane amoebas",  
	\href{http://dx.doi.org/10.1006/aima.1999.1856}{Adv. Math.} \textbf{151}(1), 
	45-70 (2000).

\bibitem{Theobald2002}
	T. Theobald, 
	"Computing amoebas", 
	\href{http://dx.doi.org/10.1080/10586458.2002.10504703}{Exp. Math.} \textbf{11}(4), 
	513-526 (2002). 

\bibitem{Mikhalkin2004}
	G. Mikhalkin,  
	"Amoebas of algebraic varieties and tropical geometry", in 
	\textit{Different faces of geometry}, \href{http://dx.doi.org/10.1007/0-306-48658-x_6}{Int. Math. Ser.} Vol. 3, 257-300
	(Eds. S. Donaldson, Y. Eliashberg and M. Gromov, Springer SBM, 2004). 

\bibitem{Tsikh2012}
	N. A. Bushueva and A. K. Tsikh AK,  
	"On amoebas of algebraic sets of higher codimension",   
	\href{http://dx.doi.org/10.1134/s0081543812080056}{Proc. Steklov Inst. Math.} \textbf{279}(1), 
	52-63 (2012). 

\bibitem{NissePassare2012}
	M. Nisse and M. Passare,  
	"Amoebas and coamoebas of linear spaces",  
	arXiv preprint, 
	\url{arXiv:1205.2808}  
	(2012). 
	
\bibitem{Mikhalkin2015} 
	G. Mikhalkin, 
	"Amoebas of half-dimensional varieties",   
	arXiv preprint, 
	\url{arXiv:1412.4658v2} 
	(2015).  

\bibitem{Mikhalkin2013} 
	G. Mikhalkin,  
	"Geometry of amoebas",  
	Lecture notes, 
	Inst. Henri Poincaré 
	(2013). 

\bibitem{Kenyon2006}
	R. Kenyon, A. Okounkov and S. Sheffield,  
	"Dimers and amoebae", 
	\href{http://dx.doi.org/10.4007/annals.2006.163.1019}{Ann. Math.} \textbf{163}(3), 
	1019-1056 
	(2006). 

\bibitem{EKR1961}
	P. Erd\H{o}s, C. Ko and R. Rado,  
	"Intersection theorems for systems of finite sets" 
	\href{http://dx.doi.org/10.1093/qmath/12.1.313}{Q. J. Math.} \textbf{12}(1), 
	313-320 (1961) 
	
\bibitem{Maclagan2015}
  	D. Maclagan and B. Sturmfels, 
  	\textit{"Introduction to Tropical Geometry"}, 
  	Graduate Studies in Mathematics Vol. 161 
  	(Providence, RI: Am. Math. Soc., 2015). 
  	
\bibitem{KodamaWilliams2013}
	Y. Kodama and L. Williams, 
	"The Deodhar decomposition of the Grassmannian and the regularity of KP solitons",  
	\href{https://doi.org/10.1016/j.aim.2013.06.011}{Advances in Mathematics} \textbf{244}, 
	979-1032  
	(2013). 

\bibitem{AK2015}
	M. Angelelli and B. Konopelchenko,   
  	"Tropical Limit in Statistical Physics",   
  	\href{http://dx.doi.org/10.1016/j.physleta.2015.04.003}{Phys. Lett. A} \textbf{379}(24-25), 
  	1497-1502 
  	(2015). 

\end{thebibliography}
\end{document}